\documentclass{sig-alternate} 

\usepackage[english]{babel}    
\usepackage{graphicx}
\usepackage{pdfpages}
\usepackage{soul}
\usepackage[small]{caption}
\makeatletter
\newtheorem{theorem}{Theorem}
\newtheorem{lemma}[theorem]{Lemma}
\newtheorem{definition}{Definition}
\newtheorem{remark}{Remark}
\newtheorem{problem}{Problem}
\newtheorem{assumption}{Assumption}
  \makeatletter
  \let\@copyrightspace\relax
  \makeatother

\newcommand{\AP}{\Sigma}
\newcommand{\define}{\triangleq}
\renewcommand{\cdots}{\dots}
\renewcommand{\ldots}{\dots}
\newcommand{\round}{T_{\circlearrowleft_m}}
\newcommand{\prob}{\mathrm{prob}}
\newcommand{\cond}{\mathrm{cond}}

\newcommand{\tl}{\texttt{tl}}
\newcommand{\tr}{\texttt{tr}}
\newcommand{\bl}{\texttt{bl}}
\newcommand{\br}{\texttt{br}}

\begin{document}

\conferenceinfo{HSCC}{'15 Seattle, USA }
 \title{Communication-Free Multi-Agent Control under Local Temporal Tasks and Relative-Distance Constraints}
\numberofauthors{3}
\author{
\alignauthor
Meng Guo\\
       \email{mengg@kth.se}
\alignauthor 
Jana Tumova\\
       \email{tumova@kth.se}
\alignauthor
Dimos V. Dimarogonas \\
       \email{dimos@kth.se}
\end{tabular}
\begin{tabular}{c}
\affaddr{ACCESS Linnaeus Center}\\
\affaddr{Center for Autonomous Systems}\\
\affaddr{Automatic Control Laboratory}\\
\affaddr{KTH - Royal Institute of Technology}\\
\affaddr{Osquldas v\"ag 10, SE-100 44, Stockholm, Sweden}
}

\maketitle 
\begin{abstract}
We propose a distributed control and coordination strategy for multi-agent systems where each agent has a local task specified as a Linear Temporal Logic (LTL) formula and at the same time is subject to relative-distance constraints with its neighboring agents. 
The local tasks capture the temporal requirements on individual agents' behaviors, while the relative-distance constraints impose requirements on the collective motion of the whole team. 
The proposed solution relies only on relative-state measurements among the neighboring agents without the need for explicit information exchange. 
It is guaranteed that the local tasks given as syntactically co-safe or general LTL formulas are fulfilled and the relative-distance constraints are satisfied at all time. The  approach is demonstrated with computer simulations.
\end{abstract}
\keywords{Hybrid Systems;
Multi-Agent Systems;
Autonomous Agents;
Formal Methods in Control;
Potential Fields;
LTL.}
\section{Introduction}

Cooperative control of multi-agent systems generally focuses on designing local control laws to achieve a  global goal, such as reference-tracking~\cite{6595620}, consensus~\cite{Ren05}, or formation~\cite{ji2007distributed}. In addition to these objectives, various relative-motion constraints are often imposed to achieve stability, safety and integrity of the overall system, such as collision avoidance~\cite{dimarogonas2006feedback}, network connectivity~\cite{ji2007distributed,5948318}, or relative velocity constraint~\cite{6595620}. 
This work is motivated by the desire to specify and achieve more structured and complex team behaviors than the listed ones. Particularly, following a recent trend, we consider Linear Temporal Logic (LTL) formulas as suitable descriptions of desired high-level goals. LTL allows to rigorously specify various temporal tasks, including periodic surveillance, sequencing, request-response, and their combinations. Furthermore, with the use of formal verification-inspired methods, a discrete plan that guarantees the specification satisfaction can be automatically synthesized, while various abstraction techniques bridge the continuous control problem and the discrete planning one. As a result, a generic hierarchical approach that allows for correct-by-design control with respect to the given LTL specification has been formulated and largely employed during the last decade or so in single-agent as well as multi-agent settings  \cite{Din11, Fil12, Guo-icra14, Klo11, loizou-cdc2005,  Lyg13, saha14,  Jana14, Ulu13}.



In temporal logic-based multi-agent  control, two different points of view can be taken: a top-down and a bottom-up. In the former one, a global specification captures requirements on the overall team behavior. Typically, the focus of decentralization is on  decomposing the specification into tasks to be executed by the individual agents in a  synchronized~\cite{Din11} or partially synchronized~\cite{Klo11, Ulu13} manner. A central monitoring unit then ensures that the composition of the local plans yields the satisfaction of the global goal. 

In contrast, in bottom-up approach, each agent is assigned its own local task. The tasks can be independent ~\cite{Fil12, Guo-icra14} or partially dependent, involving requests for collaboration with the others~\cite{Jana14}. A major research interest here is the decentralization of  planning and control procedures. For instance, in~\cite{Guo-icra14}, a decentralized revision scheme is suggested for a team in a partially-known workspace. In~\cite{Fil12}, gradual verification is employed to ensure that independent LTL tasks remain mutually satisfiable while avoiding collisions.  In~\cite{Jana14}, a receding horizon approach is employed to achieve partially decentralized planning for  collaborative tasks. In~\cite{saha14}, the authors propose a compositional motion planning
framework for multi-robot systems under safe LTL specifications.




In this work, we tackle the multi-agent control problem under local LTL tasks from the bottom-up perspective. Even though the local tasks are mutually independent, the agents within a multi-agent group are often more than a collection of stand-alone systems. Instead, they are subject to dynamic constraints with their neighboring agents and in such a case, integration of the continuous motion control with the high-level discrete network structure control is essential~\cite{6595620,5948318}. Here, the agents are subject to relative-distance constraints that need to be satisfied at all times. This coupling constraints make the team of agents {competitive} as each agent has to satisfy its local task and at the same time {cooperative} as they have to maintain the relative distance within the team.
We addressed a version of this problem in~\cite{Guo-cdc14}, where we proposed a dynamic leader-follower coordination and control scheme as a solution. In this work, however, we aim for a fully decentralized and communication-free solution that is applicable, e.g.,  to low-cost robotic systems equipped with range and angle sensors, but without communication units. 
Our solution consists of three ingredients: a~standard discrete plan synthesis algorithm, a decentralized, hybrid, potential-field-based motion controller with two different control modes and a switching strategy between the two different continuous control modes.

In summary, we propose a fully communication-free decentralized hybrid control scheme for multi-agent systems under both complex high-level local LTL tasks and low-level relative-distance constraints. Specifically, our main contribution is the  design of a two-mode communication-free control law that brings a group of agents to a region of interest.

The rest of the paper is organized as follows. Sec.~\ref{sec:prelims}  introduces  preliminaries. Sec.~\ref{sec:pf}  formalizes the considered problem. Sec.~\ref{sec:solution} presents our solution in details. Sec.~\ref{sec:example} demonstrates the results in simulations. We conclude in~Sec.~\ref{sec:conc}.

\section{Preliminaries}\label{sec:prelims}

{\emph{Linear Temporal Logic (LTL) formula} over a set of \emph{atomic propositions} $\AP$ that can be evaluated as true or false is defined inductively according to the following rules:}
{
\begin{itemize}\itemsep-0.5ex
\item an atomic proposition $\sigma \in \AP$ is an LTL formula;
\item if $\varphi$ and $\psi$ are LTL formulas, then also $\neg \varphi$, $\varphi \wedge \psi$, $\bigcirc \varphi$, $\varphi \, \textsf U \, \psi$, $\Diamond \, \varphi$, and $\square \, \varphi$ are LTL formulas, 
\end{itemize}
where $\neg $ (\emph{negation}), $\wedge$ (\emph{conjunction}) are standard Boolean connectives and $\bigcirc$ (\emph{next}), $\textsf U$ (\emph{until}), $\Diamond$ (\emph{eventually}), and $\square$~(\emph{always}) are temporal operators.}
The semantics of LTL is defined over the infinite words over~$2^\AP$. Informally, $\sigma \in \AP$ is satisfied on a word $w = w(1)w(2)\ldots$ if $\sigma \in w(1)$. Formula $\bigcirc \, \varphi$ holds true if $\varphi$ is satisfied on the word suffix that begins in the next position $w(2)$, whereas $\varphi_1 \, \textsf{U}\, \varphi_2$ states that $\varphi_1$ has to be true until $\varphi_2$ becomes true. Finally, $\Diamond \, \varphi$ and $\square \, \varphi$ are true if $\varphi$ holds on $w$ eventually and always, respectively. {For full details, see e.g.,~\cite{Bai08}.}

{\emph{Syntactically co-safe LTL (sc-LTL)} is a subclass of LTL built without the \emph{always} operator $\Box$ and with the restriction that the \emph{negation} $\neg$ can be applied to atomic propositions only~\cite{Kup01}. In contrast to general LTL formulas, the satisfaction of an sc-LTL time can be achieved in a finite time, i.e., each word satisfying an sc-LTL formula $\varphi$ consists of a \emph{satisfying prefix} that can be followed by an arbitrary suffix.}{}

\medskip

{\emph{An undirected weighted graph} 
is a tuple $G=(\mathcal{N},{E}, {W})$}, where 
$\mathcal{N}=\{1,\ldots,N\}$ is a set of {nodes};
${E} \subseteq \mathcal{N}\times \mathcal{N}$ is a set of \emph{edges}; and
${W}:{E}\rightarrow \mathbb{R}^+$ is the weight {function}.
Each node~{$i$} has a {set of \emph{neighbors}} $\mathcal{N}_i=\{j\in \mathcal{N}\,|\,(i,j)\in {E}\}$.
A path from node $i$ to $j$ is a sequence of nodes starting with $i$ and ending with $j$ such that the consecutive nodes are neighbors. 
$G$ is \emph{connected} if there is a path between any two nodes and $G$ is \emph{complete} if $E=\mathcal{N}\times \mathcal{N}$. 
The Laplacian matrix  $\textbf{H}$ of $G$ is an $N\times N$ positive semidefinite matrix:
$\textbf{H}(i,i)=\sum_{j\in \mathcal{N}_i}W(i,j), \forall i\in \mathcal{N}$; 
 $\textbf{H}(i,j)=W(i,j)$, $\forall (i,j)\in {E}$, and  $\textbf{H}(i,j)=0$ otherwise. 
{For a connected graph $G$,}  $\textbf{H}$ has nonnegative eigenvalues~\cite{Godsil} and a single zero eigenvalue with the  eigenvector $\textbf{1}_N$, where $\textbf{1}_N=[1,\ldots,1]^T$. 

\medskip

{In this paper, each vector norm over $\mathbb{R}^n$ is the Euclidean norm~\cite{horn2012matrix}}. 
{We use $|S|$ to denote the cardinality of a set $S$ and $v[i]$ to denote the $i$-{th} element of a vector $v$.}

\section{Problem Formulation}\label{sec:pf}
\subsection{Agent Dynamics and Network Structure}\label{system}

We consider a team of $N$ autonomous agents with unique identities (IDs) {$i\in\mathcal{N}=\{1,\ldots,N\}$}. They all satisfy the single-integrator dynamics: 
\begin{equation}\label{dynamics}
\dot{x}_i(t) \triangleq u_i(t), \qquad i\in \mathcal{N}
\end{equation}
where $x_i(t), \, u_i(t) \in \mathbb{R}^2$ are the {respective} state and the control input of agent~$i$ at time $t> 0${, and} $x_i(0)$ is the given {initial} state. The agents are modeled as point masses without volume, {i.e.,} inter-agent collisions are not considered.

Each agent has a sensing radius $r>0$, which is assumed to be identical for all agents. Namely, each agent can only observe 
another agent' state if their relative distance is less than $r$. 
Thus, given $\{x_i(0), i\in \mathcal{N}\}$, we 
define the undirected graph 
$G_0\triangleq(\mathcal{N}, E_0)$, where 
$(i,\,j)\in E_0$ if $\|x_i(0)-x_j(0)\|<r$. 
{We assume that the initial graph $G_0$ is connected.}

\subsection{{Task Specifications}}\label{taskspec}
Within the 2D workspace, each agent $i\in \mathcal{N}$ has a set of  {$M_i$} regions of interest: $\Pi_i\triangleq\{\varpi_{i1},\ldots,\varpi_{iM_i}\}$.
These regions can be  {of} different shapes, such as spheres, triangles, 
or polygons. For simplicity of presentation, $\varpi_{i\ell}\in \Pi_i$ is here {represented by} a circular area around a point of interest:
\begin{equation}\label{regions}
\varpi_{i\ell} = \mathcal{B}(c_{i\ell}, r_{i\ell})=\{y\in \mathbb{R}^2 \mid \|y-c_{i\ell}\|\leq r_{i\ell}\},
\end{equation}
where $c_{i\ell}\in \mathbb{R}^2$ is the center;
$r_{i\ell}\geq r_{\min}$ is the radius, and $r_{\min} > 0$ is a given minimal radius for all regions. 
{We assume that the regions of interest do not intersect  and that the workspace is bounded, which {imply the following assumption} necessary for the design of the agents' controllers:}
\begin{assumption}\label{region-assump}
(I) $\|c_{i\ell_i}-c_{j\ell_j}\|>2\,r_{\min}$, $\forall i,\, j\in \mathcal{N}$, $\forall \varpi_{i\ell_i} \in \Pi_i $ and $\forall \varpi_{j\ell_j} \in \Pi_j$. 
(II) $\|c_{i\ell}\|<c_{\max}$, $\forall i\in \mathcal{N}$ and $\forall \varpi_{i\ell}\in \Pi_i$, where $c_{\max}>0$ is a given constant. 
\end{assumption}

{Each region of interest is associated with a subset of atomic propositions $\Sigma_i$ through
the labeling function $L_i:\Pi_i\rightarrow 2^{\AP_i}$. }
Without loss of generality, we assume that $\AP_i \cap \AP_j = \emptyset$, for all $i,j \in \mathcal N$ such that $i \neq j$. 
{We view the atomic propositions $L_i(\varpi_{i\ell})$ as \emph{services} that the agent $i$ can provide when being present in region $\varpi_{i\ell} \in \Pi_i$. Hence,} 
{upon the visit to} $\varpi_{i\ell}$, the agent~$i$ {chooses among $L_i(\varpi_{i\ell})$ the subset of atomic propositions to be evaluated as true (i.e., the subset of services it provides among the available ones).
{We denote} by $\mathbf{x}_i(T)$ the \emph{trajectory} of agent $i$ during the time interval {$[0,T)$}, where $T>0$ and $T$ can be infinity. 
{The trajectory $\mathbf{x}_i(T)$ is associated with a unique {finite or infinite} sequence $\mathbf{p}_i(T) \define \pi_{i1}\pi_{i2}\dots$ of regions in $\Pi_i$ that the agent $i$ crosses, and with a finite or infinite sequence of time instants $t_{i0}'t_{i1}t_{i1}' t_{i2} t_{i2}'\cdots$ when $i$ enters/leaves the respective regions. {Formally,} for all $k \geq 1$:
$0 = t_{i0}' \leq t_{ik}\leq t_{ik}'< t_{ik+1} < T$, $x_i(t) \in \pi_{ik}$, $\pi_{ik}\in\Pi_i$, $\forall t\in [t_{ik},\,t_{ik}']$, and $x_i(t) \notin \varpi_{i\ell}$, $\forall \varpi_{i\ell} \in \Pi_i$ and $\forall t\in (t_{ik-1}',\,t_{ik})$.}
{The \emph{trace} corresponding to $\mathbf{x}_i(T)$ is a sequence of labels of the visited regions
$\textup{\texttt{trace}}_i(T) \define L_i(\pi_{i1})L_i(\pi_{i2})\cdots, $
which represents the sequence of atomic propositions that \emph{can} be true (i.e., the services that can be provided) by the agent $i$ following $\mathbf{x}_i(T)$.} 

{The \emph{word} the agent $i$ produces is a sequence of atomic propositions that actually \emph{are} evaluated as true (i.e., the actually provided services). Note that the agent's word and trajectory have to comply: if $\textup{\texttt{trace}}_i(T)$ is as above, then $\textup{\texttt{word}}_i(T) = w_{\ell_1}w_{\ell_2}\cdots$, where $w_{\ell_k} \subseteq L_i(\pi_{\ell_k})$, for all $k \geq 1$.}

\medskip


{The specification of the local task for
each agent $i\in \mathcal{N}$
is given  as a general LTL or an sc-LTL formula} $\varphi_i$ over $\AP_i$ and captures requirements on the services to be provided. In this work, we do not focus on how the service providing is executed by an agent; we only aim at controlling an agent's motion to reach regions where these services are available.
{Formally, an infinite trajectory $\mathbf{x}_i(T)$ of an agent $i$ satisfies a given formula $\varphi_i$ if and only if there exists an infinite word $\textup{\texttt{word}}_i(T)$ that complies with $\mathbf{x}_i(T)$ and satisfies $\varphi_i$.}

 
\subsection{Problem Statement}
\begin{problem}\label{main-prob}
{Given a team of $N$ agents as in Sec.~{\ref{system}}, and their task specifications as in Sec.~\ref{taskspec}, design distributed control laws $u_i$, $\forall i\in \mathcal{N}$, {such that for $T \rightarrow \infty$}:
\begin{itemize} \itemsep-0.5ex
\item[(1)] $\mathbf{x}_i(T)$ satisfies $\varphi_i$; and
\item[(2)] $\|x_i(t)-x_j(t)\|<r$, $\forall (i,\,j)\in E_0$, $\forall t\in [0,\, T)$.
\end{itemize}}


\end{problem}

\section{Solution}\label{sec:solution}

{The proposed solution consists of three layers: 
(i) 
an offline synthesis of an discrete plan, i.e., a sequence of progressive goal regions for each agent; 
(ii) 
a distributed continuous control scheme guaranteeing that one of the agents reaches its progressive goal region in finite time while the relative-distance constraints are fulfilled at all time;  
(iii) 
a hybrid control layer, which monitors the discrete plan execution and switches between different continuous control modes to achieve the satisfaction of each LTL task. }

\subsection{Discrete Plan Synthesis}\label{synthesis}

{The discrete plan can be generated using standard techniques leveraging ideas from automata-based formal verification. Loosely speaking, an LTL or an sc-LTL formula $\varphi_i$ is first translated into a B\"uchi or a finite automaton, respectively. The automaton is viewed as a graph and analyzed using graph search algorithms. As a result, a word that satisfies $\varphi_i$ is obtained and mapped onto the sequence of regions to be visited. Current temporal logic-based discrete plan synthesis algorithms can accommodate various environmental constraints and advanced plan optimality criteria. We refer the interested reader to related literature, e.g.,~\cite{Bai08,Din11}.

}


{It can be shown that without loss of generality, the derived plan of an agent $i$ is in a prefix-suffix form} 
$\tau_i = \tau_{i,\textup{pre}}(\tau_{i,\textup{suf}})^\omega$,
where {$\tau_{i,\textup{pre}}=(\pi_{i1},w_{i1}) \cdots(\pi_{i{k_i}},w_{ik_i})$ is the plan prefix, and $\tau_{i,\textup{suf}}=(\pi_{i{{k_i}+1}},w_{i{k_i}+1}) \cdots(\pi_{i{K_i}},w_{iK_i})$ is the periodical plan suffix}; 
$\pi_{ik}\in \Pi_i$ and $w_{ik}\subseteq L_i(\pi_{ik})$, $\forall k=1,\cdots, K_i$. 
Simply speaking, $\tau_i$ represents the sequence of \emph{progressive goal regions} $\pi_{i1}\pi_{i2}\cdots$ and the {word, i.e., the sequence of services to be provided there $w_{i1}w_{i2}\cdots$ that satisfies $\varphi_i$}. {If $\{\varphi_i, i\in \mathcal{N}\}$ are all sc-LTL formulas, $\tau_{i,\textup{pre}}$ represents the satisfying prefix and the suffix $\tau_{i,\textup{suf}}$ can be disregarded.}

\subsection{Continuous Controller Design}\label{continuous-design}

Before stating the proposed control scheme, let us first introduce the notion of connectivity graph, which will allow us to handle the relative-distance constraints.
Recall that each agent has a limited sensing radius $r>0$ as mentioned in Sec.~\ref{system}. Let $\delta \in (0,\,r)$ be a given constant. 
Then we define the connectivity graph $G(t)$ as follows:
\begin{definition}\label{edge}
Let $G(t)\triangleq(\mathcal{N}, E(t))$ denote the undirected time-varying connectivity graph at time $t\geq 0$, where $E(t)\subseteq \mathcal{N}\times \mathcal{N}$is the set of edges. 
(I) $G(0)=G_0$;
(II)~At time $t > 0$, $(i,\, j)\in E(t)$ iff one of the following conditions hold: (1) $\|x_i(t)-x_j(t)\|\leq r -\delta $; or (2)~$r -\delta<\|x_i(t)-x_j(t)\|\leq r $ and $(i,j) \in E(t^-)$, where $t^-<t$ and $|t-t^-|\rightarrow 0$.
\end{definition}

Note that the condition (II)  above guarantees that a new edge will only be added when the distance between two unconnected agents decreases below $r -\delta$. 
In other words, there is a hysteresis effect when adding new edges to the connectivity graph. 
Each agent $i\in \mathcal{N}$ has a time-varying set of neighbors ${\mathcal{N}}_i(t)=\{j\in \mathcal{N}\,|\,(i,\,j)\in E(t)\}$. 
Note that the   graph $G_0$ defined in Sec.~\ref{system} is assumed to be connected.

{Given that the progressive goal region at time $t$ is $\pi_{ig}= \mathcal{B}(c_{i\text{g}}, r_{i\text{g}})\in \Pi_i $}, we propose the following two control modes:

\noindent (1) the \emph{active} mode:
\begin{equation}\label{law1}
\textbf{C}_{act}: \quad u_i(t)\triangleq -\displaystyle  d_i\, p_i -\sum_{j\in \mathcal{N}_i(t)}h_{ij}\,x_{ij},
\end{equation}
(2) the \emph{passive} mode:
\begin{equation}\label{law2}
\textbf{C}_{pas}: \quad u_i(t)\triangleq -\sum_{j\in \mathcal{N}_i(t)}h_{ij}\,x_{ij},
\end{equation}
where $x_{ij}\triangleq x_i-x_j$; $p_i\triangleq x_i-c_{i\text{g}}$; 
and the coefficients $d_{i}$ and $h_{ij}$ are given by 
\begin{equation}\label{d}
d_i \triangleq \frac{{\varepsilon^3}}{(\|p_i\|^2+\varepsilon)^2}+\,\frac{\varepsilon^2}{2\,(\|p_i\|^2+\varepsilon)}; 
h_{ij}\triangleq  \frac{r^2}{(r^2-\| x_{ij}\|^2)^2},
\end{equation}
where  $\varepsilon>0$ is a design parameter to be appropriately tuned. 
We show in detail how to choose $\varepsilon$ in the sequel. 
Note that both controllers in~\eqref{law1} and~\eqref{law2} are nonlinear and rely on only locally-available states: $x_i(t)$ and $\{x_j(t),j\in \mathcal{N}_i(t)\}$.

Assume that $G(T_s)$ is connected at time $T_s>0$. Moreover, assume that there are $N_a\geq 1$ agents within $\mathcal{N}$ that are in the \emph{active} mode obeying~\eqref{law1} with its goal region as $\pi_{i\text{g}}=\mathcal{B}(c_{i\text{g}},\,r_{i\text{g}})\in \Pi_i$; and the rest $N_p=N-N_a$ agents that are in the \emph{passive} mode obeying~\eqref{law2}. 
For simplicity, denote by the group of active and passive agents $\mathcal{N}_a, \mathcal{N}_p \subseteq \mathcal{N}$ respectively. 
Note that it is allowed that $N_a=N$ and $N_p=0$ when all agents are in the active mode. 

In the rest of this section, we show that for \emph{any} allowed combination of $N_a > 1$ and $N_p < N$, by following the control laws~\eqref{law1} and~\eqref{law2},  \emph{one} active agent reaches its goal region within finite time $T_f\in (T_s,\,+\infty)$, while $\|x_i(t)-x_j(t)\|<r$, $\forall (i,\,j)\in E(T_s)$ and $\forall t\in [T_s,\,T_f]$.

\subsubsection{Relative-Distance Maintenance}\label{subsubsec:distance}
In this part, we show that the relative-distance constraints are always satisfied by following the control laws~\eqref{law1} and~\eqref{law2} for \emph{any} number of active and passive agents within the system following a potential-field based analysis.
We propose the following potential-field function:
\begin{equation}\label{Lyapunov}
V(t)\triangleq\frac{1}{2}\sum_{i\in \mathcal{N}}\sum_{j\in \mathcal{N}_i(t)} \phi_c (x_{ij} ) +b_i\sum_{i\in \mathcal{N}}\phi_g (x_{i} )
\end{equation}
where $\phi_c(\cdot)$ stands for an attractive potential to agent $i$'s neighbors and is given by:
\begin{equation}\label{potentialc}
\phi_c (x_{ij} )\triangleq\frac{1}{2}\,\frac{\| x_{ij}\|^2}{r^2-\| x_{ij}\|^2}, \qquad \|x_{ij}\|\in [0,\, r-\delta);
\end{equation}
while $\phi_g(\cdot)$ is an attractive force to agent $i$'s goal defined by:
\begin{equation}\label{potentialg}
\phi_g (x_{i} )\triangleq\frac{\varepsilon^2}{2}\,\frac{{\| p_i\|^2}}{\| p_i\|^2+\varepsilon} + \frac{\varepsilon^2}{4} \, \ln(\|p_i\|^2+\varepsilon),
\end{equation}
where function $\ln(\cdot)$ is the natural logarithm; $b_i\in \mathbb{B}$ indicates the agent $i$'s control mode. Namely, $b_i=1$, $\forall i\in \mathcal{N}_a$ and $b_i = 0$, $\forall i\in \mathcal{N}_p$. 
It can be verified that the gradient of $V(t)$ from~\eqref{Lyapunov} with respect to $x_i$ is given by 
\begin{equation}\label{gradient}
\begin{split}
\nabla_{x_i} V=\frac{\partial V}{\partial x_i}  &=  \nabla_{x_i}\phi_g (x_i )+\sum_{j\in \mathcal{N}_i} \nabla_{x_i}\phi_c (x_{ij} )\\
&= b_i\, d_i\, p_i+\sum_{j\in \mathcal{N}_i(t)} h_{ij}\,x_{ij}=-u_i.
\end{split}
\end{equation}

\begin{theorem}\label{connectivity}
$G(t)$ remains connected and no existing edges within $E(T_s)$ will be lost, namely $E(T_s)\subseteq E(t)$, $\forall t\geq T_s$.
\end{theorem}
\begin{proof}
Assume that the network $G(t)$ remains \emph{invariant} during the time period $[t_1,\, t_2)\subseteq [T_s,\,\infty)$. 
Thus the neighboring sets $\{\mathcal{N}_i, i\in \mathcal{N}\}$ also remain invariant and $V(t)$ is differentiable for $t\in [t_1,\, t_2)$. 
Then the time derivative of $V(t)$ is given by 
\begin{equation}\label{derivative}
\begin{split}
\dot{V}(t)&=\sum_{i\in \mathcal{N}} \big( \nabla_{x_i} V\big)^T \, \dot{x}_i=\sum_{i\in \mathcal{N}} \big( \nabla_{x_i} V\big)^T \, u_i\\
&=-\sum_{i\in \mathcal{N}} \big{\|}b_i\,d_i\, p_i +\sum_{j\in \mathcal{N}_i(t)}h_{ij}\, x_{ij}  \big{\|}^2\leq 0,
\end{split}
\end{equation}
meaning that $V(t)$ is non-increasing, $\forall t\geq T_s$. Thus $V(t)\leq V(T_s)<+\infty$ for $t\geq T_s$. 

On the other hand, assume a \emph{new} edge $(p,\,q)$ is added to $G(t)$ at $t=t_2$, where $p,\, q\in \mathcal{N}$. 
By Def.~\ref{edge}, $\|x_{pq}(t_2)\|\leq r-\delta$ and $\phi_c(x_{pq}(t_2))=\frac{(r-\delta)^2}{\delta(2r-\delta)}<+\infty$ since $0<\varepsilon<r$. 
Denote by $\widehat{E}\subset \mathcal{N}\times \mathcal{N}$ the set of newly-added edges at $t=t_2$. Let $V(t_2^+)$ and $V(t_2^-)$ be the value of $V(t)$ before and after adding the set of new edges to $G(t)$ at $t=t_2$. We get
$
V(t_2^+) = V(t_2^-) + \sum_{(p,\,q)\in \widehat{E}}\phi_c(x_{pq}(t_2))
\leq  V(t_2^-) + |\widehat{E}|\, \frac{(r-\delta)^2}{\varepsilon(2r-\delta)}<+\infty, 
$
where we use the fact that $|\widehat{E}|$ is bounded as $\widehat{E}\subset \mathcal{N}\times \mathcal{N}$. Thus $V(t)<+\infty$ also holds when new edges are added. 
Similar analysis can be found in~\cite{ji2007distributed}.

As a result, $V(t)<+\infty$ for $t\in [T_s,\, \infty)$. By Def.~\ref{edge}, one existing edge $(i,\,j)\in E(t)$ will be lost only if $x_{ij}(t)=r$. It implies that $\phi_c(x_{ij})\rightarrow +\infty$, i.e., $V(t)\rightarrow +\infty$ by~\eqref{Lyapunov}. By contradiction,  we can conclude that new edges might be added but no existing edges will be lost, namely $E(T_s)\subseteq E(t)$, $\forall t\geq T_s$. If $G(T_s)$ is connected, then $G(t)$ remains connected for $\forall t\geq T_s$. It completes the proof. 
\end{proof}
Note that Theorem~\ref{connectivity} holds also when $N_a=0$, i.e., there are no active agents, as \eqref{derivative} still holds  when $b_i=0$, $\forall i\in \mathcal{N}$.

\subsubsection{Convergence Analysis}\label{subsubsec:convergence}
In this part, we aim at analyzing the convergence properties of the closed-loop system. 
Since we have shown that $V(t)$ is non-increasing for all $t>T_s$ by Theorem~\ref{connectivity} above, by LaSalle's invariance principle~\cite{khalil2002nonlinear} we only need to find out the largest invariant set that $\dot{V}(t)=0$, which implies:
\begin{equation}\label{equi}
b_i\, d_i\,p_i +\sum_{j\in \mathcal{N}_i(t)}h_{ij}\,x_{ij} =0, \quad \forall i\in \mathcal{N}.
\end{equation}
Then we can construct one $N\times N$ diagonal matrix $\textbf{D}$ that $\textbf{D}(i,i)=b_i\,d_i$, $\forall i\in\mathcal{N}$ and $\textbf{D}(i,j)=0$, $i\neq j$ and $i,j\in \mathcal{N}$. 
and another  $N\times N$ matrix $\textbf{H}$ that $\textbf{H}(i,i)=\sum_{j\in \mathcal{N}_i}h_{ij}$, $\forall i\in \mathcal{N}$ and $\textbf{H}(i,j)=-h_{ij}$, $i\neq j$ and $\forall(i,\,j)\in E(t)$ while $\textbf{H}(i,j)=0$, $\forall (i,\,j)\notin E(t)$. Note that $h_{ij} > 0$ as $\|x_{ij}\|\in [0,\,r)$ by~\eqref{gradient}, $\forall (i,\,j)\in E(t)$. 
As a result, $\textbf{H}$ is the Laplacian matrix of the graph $G(t)=(\mathcal{N}, E(t), W)$, where $W(i,\,j)=h_{ij}$, $\forall (i,\,j)\in E(t)$.
Then~\eqref{equi} can be written in vector form:
\begin{equation}\label{equi2}
\textbf{H}\otimes \textbf{I}_2 \cdot \mathbf{x} + \textbf{D} \otimes \textbf{I}_2 \cdot (\mathbf{x}-\mathbf{c})=0,
\end{equation}
where $\otimes$ is the Kronecker product~\cite{horn2012matrix}; $\mathbf{x}$ is the stack vector for $x_i$, $i\in \mathcal{N}$ and $\mathbf{x}[i]=x_i$; 
$\textbf{I}_2$ is the identity matrix; $\mathbf{c}$ is the stack vector that $\mathbf{c}[i]=c_{i\text{g}}$ if $i\in \mathcal{N}_a$ and $\mathbf{c}[i]=\textbf{0}_2$ if $i\in \mathcal{N}_p$. Let $\mathcal{C}$ be the set of  critical points satisfying~\eqref{equi2}, i.e.,
\begin{equation}\label{equiset}
\mathcal{C}\triangleq\{x\in \mathbb{R}^{2N}\,|\, \textbf{H}\otimes \textbf{I}_2 \cdot \mathbf{x} + \textbf{D} \otimes \textbf{I}_2 \cdot (\mathbf{x}-\mathbf{c})=0\}.
\end{equation}
Now we show that at the critical points all agent relative distances can be made arbitrarily small by reducing $\varepsilon$ and the corresponding connectivity graph is a complete graph.

\begin{lemma}\label{difference}
For all critical points $\mathbf{x}_c\in \mathcal{C}$, (I) $\|x_{ij}\|$ can be made arbitrarily small by reducing $\varepsilon$, $\forall (i,\,j)\in E(t)$;
(II) there exists $\varepsilon_0>0$ such that if $\varepsilon<\varepsilon_0$, then the connectivity graph $G(t)$ is complete.
\end{lemma}
\begin{proof}
(I) Consider the following equation for $\mathbf{x}_c\in \mathcal{C}$
$$
\sum_{(i,j)\in {E}(t)} h_{ij} \|x_{ij}\|^2 
 = \mathbf{x}^T_c \cdot (\textbf{H}\otimes \textbf{I}_2) \cdot \mathbf{x}_c.
$$
Combining the above equation  with~\eqref{equi2}, we get  
\begin{equation}\label{relativebound}
\begin{split}
&\sum_{(i,j)\in {E}(t)} h_{ij} \|x_{ij}\|^2  =-\mathbf{x}_c^T \cdot (\textbf{D}\otimes \textbf{I}_2)\cdot (\mathbf{x}_c-\mathbf{c})\\
& =-(\mathbf{x}_c-\mathbf{c})^T \cdot (\textbf{D}\otimes \textbf{I}_2)\cdot (\mathbf{x}_c-\mathbf{c})\\
&\quad \; \; -\mathbf{c}^T \cdot (\textbf{D}\otimes \textbf{I}_2)\cdot (\mathbf{x}_c-\mathbf{c})\\
& =-\sum_{i\in \mathcal{N}}b_i\,d_i\, \big(\|p_i\|^2 + c_{i\ell}^T \, p_i\big) \leq \sum_{i\in \mathcal{N}} b_i\, \|c_{i\ell}\|\,d_i\|p_i\|.
\end{split}
\end{equation}
Since $d_i\,\|p_i\|<\varepsilon\sqrt{\varepsilon}$ for $\|p_i\|\geq 0$, we get
\begin{equation}\label{relativebound2}
\sum_{(i,j)\in {E}(t)} h_{ij} \|x_{ij}\|^2 <N_a\,c_{\max}\,\varepsilon \sqrt{\varepsilon}\leq N\,c_{\max}\,\varepsilon \sqrt{\varepsilon},
\end{equation}
where $\|c_{i\ell}\|<c_{\max}$ is given in Assump.~\ref{region-assump}.
Thus $\forall (i,\,j)\in E(t)$, it holds that $h_{ij}\|x_{ij}\|^2<N\,c_{\max}\,\varepsilon \sqrt{\varepsilon}\triangleq \varsigma$. It can be verified that $h_{ij}\|x_{ij}\|^2$ is monotonically increasing as a function of $\|x_{ij}\|$. 
This implies that $\forall (i,\,j)\in E(t)$,
$\|x_{ij}\|^2\leq r^2\,\varsigma$,  or equivalently $\|x_{ij}\|^2\leq \varepsilon \sqrt{\varepsilon}\,\xi$,
where
\begin{equation}\label{xi}
\xi \triangleq r^2\, N\, c_{\max}. 
\end{equation}
Thus $\|x_{ij}\|$ can be made arbitrarily small by reducing $\varepsilon$. 


(II) Moreover, let $\varepsilon_0$ satisfy 
\begin{equation}\label{v0}
(N-1)\sqrt{\varepsilon_0\sqrt{\varepsilon_0}\, \xi} <r-\delta.
\end{equation}
If $\varepsilon<\varepsilon_0$, then for {any} pair $(p,\,q)\in \mathcal{N}\times \mathcal{N}$, $\|x_{pq}\|$ satisfies 
$$
\|x_{pq}\|= |x_p-x_1+x_1-x_2+\ldots -x_q|\leq (N-1) \sqrt{\varepsilon\sqrt{\varepsilon}\,\xi} <r-\delta,
$$
where we use two facts: there exists a path in $G(t)$ of maximal length $N$ from any node $p\in \mathcal{N}$ to another node $q$ as $G(t)$ remains connected for $t>T_s$ by Lemma~\ref{connectivity}; and $\|x_{ij}\|\leq \varepsilon \sqrt{\varepsilon}\,\xi$ from above, $\forall (i,\,j)\in E(t)$. 
By Def.~\ref{edge} this implies $(p,\,q)\in E(t)$. Thus $G(t)$ is a complete graph.
\end{proof}

Before stating the convergence property, we need to define the following sets for all $i\in \mathcal{N}_a$:
\begin{equation}\label{si}
\mathcal{S}_i\triangleq \{\mathbf{x}\in \mathbb{R}^{2N}\,|\,\|\mathbf{x}-\mathbf{1}_N\otimes c_{i\ell}\|\leq r_S(\varepsilon)\},
\end{equation}
where $r_S(\varepsilon) \triangleq \sqrt{3N\,\varepsilon}+\sqrt{(N-1)\varepsilon\sqrt{\varepsilon}\, \xi}$ and $\xi$ is defined in~\eqref{xi}. Loosely speaking, $\mathcal{S}_i$ represents the neighbourhood around the goal region center of the active agent $i\in \mathcal{N}_a$. Furthermore, let $\mathcal{S}\triangleq\cup_{i\in \mathcal{N}_a} \mathcal{S}_i$ and $\mathcal{S}^{\neg}\triangleq\mathbb{R}^{2N} \setminus\mathcal{S}$. 

In the following, we analyze the properties of the critical points within $\mathcal{S}$ and $\mathcal{S}^{\neg}$. More specifically:
by Lemma~\ref{away} there are no local minimal but  saddle points within $\mathcal{S}^{\neg}$; by Lemma~\ref{isolated} these saddle points are non-degenerate, 
(II) by Lemmas~\ref{distance}-\ref{localmini} all critical points within $\mathcal{S}$ are  local minima.

To explore these properties, we compute the second partial derivatives of $V(t)$ with respect to $x_i$, which are given by 
\begin{equation}\label{2gradient}
\begin{split}
\frac{\partial^2 V}{\partial x_i\partial x_i }&= b_i\, d_i \otimes \textbf{I}_2 + b_i\, d_i'\, p_i\cdot p_i^T \\
&+\sum_{j\in \mathcal{N}_i(t)} \big(h_{ij}\otimes \textbf{I}_2 +  h_{ij}'\,x_{ij}\cdot x_{ij}^T\big)
\end{split}
\end{equation}
and 
\begin{equation}\label{2gradient2}
\begin{split}
\frac{\partial^2 V}{\partial x_i\partial x_j }&= - h_{ij}\otimes \textbf{I}_2 -  h_{ij}'\,x_{ij}\cdot x_{ij}^T, \qquad \forall j\neq i,
\end{split}
\end{equation}
where 
\begin{equation}\label{hp}
d_i' =\frac{{-4\,\varepsilon^3}}{(\|p_i\|^2+\varepsilon)^3}+\frac{{-\,\varepsilon^2}}{(\|p_i\|^2+\varepsilon)^2},
\text{ and }
h_{ij}'= \frac{4\,r^2}{(r^2-\| x_{ij}\|^2)^3}.
\end{equation}


\begin{lemma}\label{away}
There are no local minima of $V$ within $\mathcal{S}^\neg$.
\end{lemma}
\begin{proof}
We prove this by showing that if a critical point $\mathbf{x}_c\in \mathcal{S}^{\neg}$ there always exists a direction $\mathbf{z}\in \mathbb{R}^{2N}$ at $\mathbf{x}_c$ such that the {quadratic form} $\mathbf{z}^T \nabla^2V \mathbf{z}$ is negative semi-definite. 

Given a critical point $\mathbf{x}_c\in \mathcal{C}$ and $\mathbf{x}_c\in \mathcal{S}^\neg$, then by definition $\|\mathbf{x}-\mathbf{1}_N\otimes c_{i\ell}\|>r_S,\, \forall i \in \mathcal{N}_a$. On the other hand, for any $i\in \mathcal{N}_a$, we can bound $\|\mathbf{x}-\mathbf{1}_N\otimes c_{i\ell}\|$ as follows:
\begin{equation*}
\begin{split}
&\|\mathbf{x}-\mathbf{1}_N\otimes c_{i\ell}\|=\|\mathbf{x}-\mathbf{1}_N\otimes x_{i}+\mathbf{1}_N\otimes (x_{i}- c_{i\text{g}})\|\\
&\leq \sqrt{\sum_{j\in \mathcal{N}} \|x_{ij}\|^2} + \sqrt{N}\, \|p_i\|\leq \sqrt{(N-1)\varepsilon\sqrt{\varepsilon}\, \xi} + \sqrt{N} \,\|p_i\|,
\end{split}
\end{equation*}
where $\|x_{ij}\|^2\leq \varepsilon\, \sqrt{\varepsilon}\,\xi$ at $\mathbf{x}_c$, $\forall (i,\,j)\in E(t)$ by Lemma~\ref{difference}.
By comparing it with $r_S(\varepsilon)$, we get
$\|p_i\|\geq \sqrt{3\,\varepsilon}$, $\forall i \in \mathcal{N}_a$.

Choose $\mathbf{z} \triangleq \mathbf{1}_N\otimes {z}$, where $z\in \mathbb{R}^2$ and $\|z\|\triangleq 1$. Then $\mathbf{z}^T \, \nabla^2V\, \mathbf{z}$ is evaluated by using~\eqref{2gradient}-\eqref{hp}:
\begin{equation*}
\begin{split}
\mathbf{z}^T \nabla^2V \mathbf{z}&=\sum_{i\in \mathcal{N}} b_i\, d_{i}\,z^T z + b_i\, d_{i}' \, z^T p_i\, p_{i}^Tz\triangleq z^T M z,
\end{split}
\end{equation*}
where $M\triangleq \sum_{i\in \mathcal{N}_a}(d_i\otimes \textbf{I}_2 + d_i'\, p_i\, p_{i}^T)$ is a $2\times 2$ Hermitian matrix.
The trace of $M$ is computed as 
\begin{equation}\label{trace}
\begin{split}
\textbf{trace}(M)&= \sum_{i\in \mathcal{N}_a} 2\,d_i + d_i'\, \|p_i\|^2\\
&=\varepsilon^3\sum_{i\in \mathcal{N}_a}\frac{{3\varepsilon-\|p_i\|^2}}{(\|p_i\|^2+\varepsilon)^3}<0,
\end{split}
\end{equation}
as we have shown that $\|p_i\|\geq \sqrt{3\varepsilon}, \forall i \in \mathcal{N}_a$ if $\mathbf{x}_c \in \mathcal{S}^{\neg}$. On the other hand, denote by $p_i=[p_{i,x},\,p_{i,y}]$ the coordinates of~$p_i$. The determinant of $M$ is given by 
\begin{equation}\label{determinant}
\begin{split}
&\textbf{det}(M)=
-(\sum_{i\in \mathcal{N}_a} d_i' \, p_{i,x}\,p_{i,y})^2\\
&\qquad +(\sum_{i\in \mathcal{N}_a} d_i + d_i'\, p_{i,x}^2)(\sum_{i\in \mathcal{N}_a} d_i + d_i'\, p_{i,y}^2)\\
&\geq \frac{1}{2}\sum_{i,\,j\in \mathcal{N}_a} \big[(d_i+d_i'\|p_{i}\|^2)(d_j+d_j'\|p_{j}\|^2)\big]>0,
\end{split}
\end{equation}
since $d_i'\|p_i\|^2<-d_i$ for $\|p_i\|>\sqrt{3\varepsilon}$, $\forall i\in \mathcal{N}_a$; and $(p_{i,x}p_{i,y}-p_{j,x}p_{j,y})^2\leq \|p_i\|^2\|p_j\|^2$ by Cauchy-Schwarz inequality~\cite{horn2012matrix}.

Denote by $\lambda_1$ and $\lambda_2$ the eigenvalues of $M$, where $\lambda_{1},\lambda_2\in \mathbb{R}$ as $M$ is Hermitian. 
Since $\textbf{trace}(M)<0$ and $\textbf{det}(M)>0$, then $M$ is negative definite and both eigenvalues are negative~\cite{horn2012matrix}, i.e.,  $\lambda_1, \lambda_2<0$.  Thus for any vector $v= \textbf{1}_N\otimes z$ where $z\in \mathbb{R}^2$, $v^T\nabla^2 V v<0$. 
To conclude, for any critical point $\mathbf{x}_c\in \mathcal{C}$, if $\mathbf{x}_c\in \mathcal{S}^\neg$ then $\mathbf{x}_c$ is not a local minimum.
\end{proof}

\begin{lemma}\label{isolated}
There exists $\varepsilon_1>0$ such that if $\varepsilon<\varepsilon_1$,  all critical points of $V$ in $\mathcal{S}^{\neg}$ are non-degenerate saddle points. 
\end{lemma}
\begin{proof}
To show that $V$ is Morse we use Lemma 3.8 from~\cite{rimon1988exact}, which states that the non-singularity of a linear operator follows from the fact that its associated quadratic form is sign definite on complementary subspaces.

Let $\mathcal{Q}=\{v\in \mathbb{R}^{2N}\,|\, v=\mathbf{1}_N\otimes z, \, z \in\mathbb{R}^2\}$.
In Lemma~\ref{away}, we have shown that for any vector $v\in \mathcal{Q}$, $v^T\nabla^2 V v<0$. 
Let $\mathcal{P}=\{v\in \mathbb{R}^{2N}\,|\, v=\mathbf{e}_N\otimes z,\; \mathbf{e}_N\perp \mathbf{1}_N,\, \mathbf{e}_N\in \mathbb{R}^N,\, z \in \mathbb{R}^2\}$.  
Firstly, it can be easily verified that 
$\mathcal{P}$ is the orthogonal complement of $\mathcal{Q}$. 
In the following, we show that $\nabla^2 V$ is positive definite in $\mathcal{P}$. 
Let $\mathbf{z} \in \mathcal{P}$, i.e., $\mathbf{z}\triangleq \mathbf{e}_N \otimes z \triangleq [z_1^T\;z_2^T\ldots z_n^T]^T$, where $z\in \mathbb{R}^2$, $\mathbf{e}_N\in \mathbb{R}^N$, $\mathbf{e}_N^T \perp \mathbf{1}_N$, $z_i\in \mathbb{R}^2$, $\forall i\in \mathcal{N}$. 
The quadratic form $\mathbf{z}^T \, \nabla^2V\, \mathbf{z}$ at $\mathbf{x}_c$ is computed using~\eqref{2gradient}-\eqref{hp}:
\begin{equation*}
\begin{split}
&\mathbf{z}^T \nabla^2V \mathbf{z}=\sum_{i\in \mathcal{N}_a} \big(d_{i}\,\|z_i\|^2 + d_{i}' \, |p_{i}^Tz_i|^2\big) \\
&+\sum_{(i,\,j)\in E(t)}\big( h_{ij} \,\|z_i-z_j\|^2 +2\,h'_{ij}\,|(x_i-x_j)^T(z_i-z_j)|^2\big)\\
&\geq \sum_{i\in \mathcal{N}_a} \big(d_{i}\,\|z_i\|^2 + d_{i}' \, |p_{i}^Tz_i|^2)+\sum_{(i,\,j)\in E(t)} h_{ij} \,\|z_i-z_j\|^2  \\
&\geq \sum_{i\in \mathcal{N}_a} \big(d_{i} + d_i'\|p_i\|^2\big)\|z_i\|^2+ 
 \mathbf{z}^T (\mathbf{H}\otimes \textbf{I}_2) \mathbf{z},
\end{split}
\end{equation*}
where we use the fact that $h_{ij}'>0$, $d_i'<0$ and $|p_{i}^Tz_i|\leq \|p_i\|\|z_i\|$. It can be verified that $d_{i} + d_i'\|p_i\|^2>-0.1\varepsilon$ for $\|p_i\|\geq \sqrt{3\varepsilon}$, $\forall i\in \mathcal{N}_a$. Moreover, 
\begin{equation}\label{laplacian}
\begin{split}
\mathbf{z}^T (\mathbf{H}\otimes \textbf{I}_2) \mathbf{z}&=(\mathbf{e}_N \otimes z)^T \cdot (\mathbf{H}\otimes \textbf{I}_2)\cdot  (\mathbf{e}_N \otimes z)\\
&= (\mathbf{e}^T_N\cdot \mathbf{H}\cdot \mathbf{e}_N) \|z\|^2\geq \lambda_2(\mathbf{H})\|z\|^2, 
\end{split}
\end{equation}
where we apply the Courant-Fischer Theorem~\cite{horn2012matrix}: 
$$\min_{\mathbf{e}_N\perp \textbf{1}_N}\{\mathbf{e}^T_N\cdot \mathbf{H}\cdot \mathbf{e}_N\}=\lambda_2(\textbf{H})>0,$$
since $\textbf{H}$ is the Laplacian matrix defined in~\eqref{equi2}, which is positive semidefinite with $\lambda_1(\mathbf{H})=0$, of which the corresponding eigenvector is $\mathbf{1}_N$; and the second smallest eigenvalue $\lambda_2(\mathbf{H})>0$.
In addition, since $h_{ij}>1/{r^2}$ and $G(t)$ is a complete graph at $\mathbf{x}_c$ by Lemma~\ref{difference}, it holds that $\lambda_2(\mathbf{H})>N/{r^2}$ by~\cite{Godsil}. 
This implies that 
\begin{equation}\label{quadratic5}
\begin{split}
\mathbf{z}^T \nabla^2V \mathbf{z}
&\geq \sum_{i\in \mathcal{N}_a} \big(\frac{N}{r^2}+d_{i} + d_i'\|p_i\|^2\big)\|z_i\|^2\\
&\geq \sum_{i\in \mathcal{N}_a} \big(\frac{N}{r^2}-0.1\varepsilon\big)\|z_i\|^2.
\end{split}
\end{equation}
Thus if 
$\varepsilon < {N}/({0.1 r^2})$, 
it holds that $\mathbf{z}^T \nabla^2V \mathbf{z}>0$, $\forall \mathbf{z}=\mathbf{e}_N \otimes z$ where $\mathbf{e}_N\perp \mathbf{1}_N$, $z\in \mathbb{R}^2$.

To conclude, $\nabla^2V|_{\mathcal{Q}}$ is negative definite by Lemma~\ref{away} and $\nabla^2V|_{\mathcal{P}}$ is positive definite by the analysis above. 
By applying Lemma 3.8 from~\cite{rimon1988exact}, we can conclude that $\nabla^2V$ is non-singular at the saddle points $\mathbf{x}_c\in \mathcal{S}^{\neg}$, if 
\begin{equation}\label{v1}
\varepsilon<\min\{\varepsilon_0,\, \frac{N}{0.1 r^2}\}\triangleq \varepsilon_1.
\end{equation}
In other words, all critical points within $\mathcal{S}^{\neg}$ are non-degenerate saddle points if $\varepsilon<\varepsilon_1$.
\end{proof}

Now we focus on proving that all critical points within $\mathcal{S}$ are stable local minima. First of all, we need the following two lemmas to show that when a critical point belongs to $\mathcal{S}_i$ corresponding to one active agent $i\in \mathcal{N}_a$, then all the other agents are within its goal region $\pi_{i\text{g}}$ and away from their own goal region center by at least distance $r_{\min}$.
\begin{lemma}\label{distance}
There exists $\varepsilon_2>0$, such that if $\varepsilon<\varepsilon_2$, the following  hold: 
(I) 
$\mathcal{S}_i\cap \mathcal{S}_j=\emptyset$, $\forall i\neq j$ and $i,\,j \in \mathcal{N}_a$; 
(II)
If $\mathbf{x}_c\in \mathcal{S}_{i^\star}$ for some $i^\star \in \mathcal{N}_a$, then $x_j\in \pi_{i^\star\textup{g}}$, $\forall j\in \mathcal{N}$ and $\|x_j-c_{j\textup{g}}\|>r_{\min}$, $j\neq i^\star$, $\forall j\in \mathcal{N}_a$.
\end{lemma}
\begin{proof}
Let $\varepsilon_2$ be given as the solution of 
\begin{equation}\label{v2}
r_S(\varepsilon_2)=\sqrt{3N\,\varepsilon_2}+\sqrt{(N-1)\varepsilon_2\sqrt{\varepsilon_2}\, \xi}\triangleq r_{\min},
\end{equation}
where $r_{\min}$ is given in Assump.~\ref{region-assump}. Note that~\eqref{v2} has an unique solution as the left-hand side is a function of $\varepsilon_2$ that monotonically increases and has the range $[0,\;\infty)$.

Assume that $\mathbf{x}_c \in \mathcal{S}_{i^\star}$ for $i^\star\in \mathcal{N}_a$, i.e., $\|\mathbf{x}_c-\mathbf{1}_N\otimes c_{i^\star\text{g}}\|\leq r_S(\varepsilon_2)$. 
Then $\forall j\neq i^\star$, $j\in \mathcal{N}_a$, it holds that
(I) 
$\|\mathbf{x}_c-\mathbf{1}_N\otimes c_{j\text{g}}\|=\|\mathbf{x}_c-\mathbf{1}_N\otimes c_{i\text{g}}+\mathbf{1}_N\otimes c_{i \text{g}}-\mathbf{1}_N\otimes c_{j\text{g}}\|
\geq \sqrt{N}\, \|c_{i\text{g}}-c_{j\text{g}}\|-\|\mathbf{x}_c-\mathbf{1}_N\otimes c_{i\text{g}}\|
\geq 2\sqrt{N} \,{r}_{\min} -r_S(\varepsilon),$
due to that $\|c_{i\text{g}}-c_{j\text{g}}\|>2r_{\min}$ by Assump.~\ref{region-assump}.
Since $\varepsilon<\varepsilon_2$, then $r_S(\varepsilon)<r_S(\varepsilon_2)=r_{\min}$. 
Thus $\|\mathbf{x}_c-\mathbf{1}_N\otimes c_{j\text{g}}\|>2\sqrt{N} \,{r}_{\min}-r_{\min}>r_{\min}=r_S(\varepsilon_2)$, implying that $\mathbf{x}_c \notin \mathcal{S}_j$.
(II) 
$\|x_j-c_{i^\star\text{g}}\|<\|\mathbf{x}_c-\mathbf{1}_N\otimes c_{i^\star\text{g}}\|<r_{\min}<r_{i^\star\text{g}}$, meaning that $x_j\in \pi_{i^\star\text{g}}$, $\forall j\in \mathcal{N}$. 
Moreover, $
\|x_j- c_{j\text{g}}\|=\|x_j- c_{i^\star\text{g}}+c_{i^\star\text{g}}-c_{j\text{g}}\|\geq \|c_{i^\star\text{g}}-c_{j\text{g}}\|-\|x_j- c_{i^\star\text{g}}\|\geq 2r_{\min}-r_{\min}>r_{\min}.
$
\end{proof}
\begin{lemma}\label{furtherbound}
There exists $\varepsilon_6>0$ such that if $\varepsilon<\varepsilon_6$, then for a critical point $\mathbf{x}_c \in \mathcal{S}_{i^\star}$, $i^\star\in \mathcal{N}_a$, then it holds that $\|p_{i^\star}\|<\sqrt{0.4\varepsilon}$.
\end{lemma}
\begin{proof}
By summing~\eqref{equi} for all $i\in \mathcal{N}$, we get 
\begin{equation}\label{distancecond}
d_{i^\star}\, p_{i^\star}=-\sum_{j\neq i^\star, j\in \mathcal{N}_a} d_j\, p_j.
\end{equation}
Consider the scalar function $f(\|p_j\|)=d_j(\|p_j\|)\|p_j\|$ for $\|p_j\|\geq 0$.
It is monotonically increasing for $\|p_j\|\in [0, \, 3.2\sqrt{\varepsilon})$ and decreasing for $\|p_j\|\in [3.2\sqrt{\varepsilon},\, \infty)$.

If $\mathbf{x}_c \in \mathcal{S}_{i^\star}$ for $i^\star\in \mathcal{N}_a$, then $\|\mathbf{x}_c-\mathbf{1}_N\otimes c_{i^\star\text{g}}\|\leq r_S(\varepsilon_2)$. 
Moreover,
$\|\mathbf{x}-\mathbf{1}_N\otimes c_{i^\star \ell}\|\geq 
\|\mathbf{1}_N\otimes x_{i^\star}-\mathbf{1}_N\otimes c_{i^\star\text{g}}\|-
\|\mathbf{x}-\mathbf{1}_N\otimes x_{i^\star}\|
\geq \sqrt{N} \,\|p_{i^\star}\|- \sqrt{(N-1)\varepsilon\sqrt{\varepsilon}\, \xi}.$
This implies
$\|p_{i^\star}\|\leq \sqrt{3\,\varepsilon}+2\sqrt{\varepsilon\sqrt{\varepsilon}\, \xi}.$
Moreover by Lemma~\ref{distance},  $\|p_j\|>r_{\min}$ , $\forall j\neq i^\star$, $j\in \mathcal{N}_a$. 
Thus if $r_{\min}>3.2\sqrt{\varepsilon}$, namely 
\begin{equation}\label{v3}
\varepsilon<0.07\, r^2_{\min}\triangleq \varepsilon_3,
\end{equation}
it holds that $d_j\,\|p_j\|<0.5\varepsilon^2/r_{\min}$, $\forall j\neq i^\star$, $j\in \mathcal{N}_a$. 
Thus $d_{i^\star}\, \|p_{i^\star}\|<0.5(N_a-1)\varepsilon^2/r_{\min}$ by~\eqref{distancecond}. 
If the following two conditions hold:
(i) $\sqrt{3\,\varepsilon}+2\sqrt{\varepsilon\sqrt{\varepsilon}\, \xi}<3.2\sqrt{\varepsilon}$;
(ii)
$0.5(N_a-1)\varepsilon^2/r_{\min}< d_j(\sqrt{0.4\varepsilon})\sqrt{0.4\varepsilon}$, 
then $\|p_{i^\star}\|<\sqrt{0.4\varepsilon}$ since it is shown earlier that function $d_j(\|p_j\|)\|p_j\|$ is monotonically increasing for $\|p_j\|\in [0,\, 3.2\sqrt{\varepsilon})$. Condition~(i) above implies that
$
\varepsilon<4.1/\xi^2\triangleq \varepsilon_4
$
and condition~(ii) holds for all $N_a\leq N$ if
$
\varepsilon<0.8\,r_{\min}^2/(N-1)^2\triangleq \varepsilon_5.
$
To conclude, if $\varepsilon<\varepsilon_6$, where 
\begin{equation}\label{v6}
\varepsilon_6\triangleq \min\{\varepsilon_3, \varepsilon_4, \varepsilon_5\},
\end{equation}
then $\mathbf{x}_c\in \mathcal{S}_{i^\star}$ implies $\|p_{i^\star}\|<\sqrt{0.4\varepsilon}$.
\end{proof}
With the above two lemmas, we can now show that all critical points within $\mathcal{S}$ are stable local minima.
\begin{lemma}\label{localmini}
There exists $\varepsilon_{\min}>0$ such that if $\varepsilon<\varepsilon_{\min}$, 
all critical points of $V$ within $\mathcal{S}$ are  local minima.
\end{lemma}
\begin{proof}
A critical point $\mathbf{x}_c\in \mathcal{S}$ can only belong to one set $\mathcal{S}_i$ for $i\in \mathcal{N}_a$ by Lemma~\ref{distance}. Without loss of generality, let $\mathbf{x}_c \in S_{i^{\star}}$, where $i^\star \in \mathcal{N}_a$. In the following we show that $\mathbf{x}_c$ is a  local minimum.

Let $\mathbf{z}\in \mathbb{R}^{2N}$ and $\|\mathbf{z}\|=1$. Set $\mathbf{z}=[z_1^T\;z_2^T\ldots z_n^T]^T$, where $z_i\in \mathbb{R}^2$, $\forall i\in \mathcal{N}$. Then $\mathbf{z}^T \, \nabla^2V\, \mathbf{z}$ at $\mathbf{x}_c$ is computed as:
\begin{equation}\label{quadratic3}
\begin{split}
&\mathbf{z}^T \nabla^2V \mathbf{z}=\sum_{i\in \mathcal{N}_a} \big(d_{i}\,\|z_i\|^2 + d_{i}' \, |p_{i}^Tz_i|^2\big)+ \\
&\sum_{(i,\,j)\in E(t)}\big( h_{ij}\|z_{ij}\|^2 +2\,h'_{ij}\,|x_{ij}^T\,z_{ij}|^2\big).
\end{split}
\end{equation}
where $z_{ij}\triangleq z_i-z_j$. 
Since $|p_i^Tz_i|\leq \|p_i\|\|z_i\|$, $d_i>0$ and $d_i'<0$, it holds that
$
d_{i}\,\|z_i\|^2 + d_{i}' \, |p_{i}^T\,z_i|^2\geq (d_i +d_{i}' \|p_i\|^2)\|z_i\|^2,  \forall i\in \mathcal{N}_a.
$
It can be verified that for $j\neq i^\star$ and $\forall j\in \mathcal{N}_a$, $d_j+d_j'\|p_j\|^2>\varepsilon^2\hat{g}$ where $\hat{g} \triangleq  -2/r^2_{\min}$, since $\|p_j\|>r_{\min}$ by Lemma~\ref{distance}; and $d_{i^\star}+d_{i^\star}'\|p_{i^\star}\|^2>0.08\varepsilon $ since $\|p_{i^\star}\|>\sqrt{0.4\varepsilon}$ by Lemma~\ref{furtherbound}. Regarding the second term of~\eqref{quadratic3}, since Lemma~\ref{difference} shows that $G(t)$ is a complete graph at $\mathbf{x}_c$ with $h_{ij}>1/{r^2}$ and $h_{ij}'>0$, we get
\begin{equation}\label{hbound}
\begin{split}
&\sum_{(i,\,j)\in E}\big( h_{ij} \,\|z_{ij}\|^2 +2\,h'_{ij}\,|x_{ij}^T\, z_{ij}|^2\big)\\
&\geq \sum_{j\in \mathcal{N}} h_{i^\star j} \,\|z_{i^\star j}\|^2 \geq \frac{1}{r^2}\sum_{j\in \mathcal{N}}\,\|z_{i^\star j}\|^2.
\end{split}
\end{equation}
Thus ~\eqref{quadratic3} can be bounded by 
\begin{equation*}\label{quadratic4}
\begin{split}
&\mathbf{z}^T \nabla^2V \mathbf{z}
\geq \sum_{i\in \mathcal{N}_a} \big(d_{i} + d_{i}' \, \|p_{i}\|^2\big)\|z_i\|^2 +\sum_{j\in \mathcal{N}} h_{i^\star j} \,\|z_{i^\star j}\|^2\\
&\geq 0.08\,\varepsilon \|z_{i^\star}\|^2 -\varepsilon^2 \sum_{j\neq i^\star,j\in \mathcal{N}_a}|\hat{g}|\|z_j\|^2+\frac{1}{r^2} \sum_{j\in \mathcal{N}}  \,\|z_{i^\star j}\|^2\\
&\geq  \sum_{j\in \mathcal{N}_a}\big(\frac{1}{r^2}+ \frac{0.08\varepsilon}{N}\big) \|z_{i^\star}\|^2 +\big(\frac{1}{r^2}-\varepsilon^2|\hat{g}|\big)\|z_j\|^2-\frac{2}{r^2} z_{i^\star}^T\, z_j,
\end{split}
\end{equation*}
as $1\leq N_a\leq N$. If the following condition holds:
\begin{equation}\label{condition11}
\begin{split}
\big(\frac{1}{r^2}+ \frac{0.08\varepsilon}{N}\big) \big(\frac{1}{r^2}-\varepsilon^2|\hat{g}|\big)> \big(\frac{1}{r^2}\big)^2,
\end{split}
\end{equation}
it implies 
$
\mathbf{z}^T \nabla^2V \mathbf{z}> (|z_{i^\star}^T\, z_j|-z_{i^\star}^T\, z_j)/r^2\geq 0,
$
$\forall \mathbf{z}\in \mathbb{R}^{2N}$.
Namely, $\nabla^2V$ is positive definite at critical points $\mathbf{x}_c \in \mathcal{S}$. 
Condition~\eqref{condition11} is equivalent to 
$$
\varepsilon^2+\frac{N}{0.08\,r^2}\varepsilon -\frac{1}{r^2|\hat{g}|}<0.
$$
Since $\varepsilon>0$, this implies that 
\begin{equation}\label{v7}
0<\varepsilon <\frac{\sqrt{(\frac{N}{0.08\,r^2})^2+\frac{4}{r^2|\hat{g}|}}-\frac{N}{0.08\,r^2}}{2}\triangleq \varepsilon_7,
\end{equation}
To conclude, if 
\begin{equation}\label{v8}
\varepsilon<\min\{\varepsilon_1, \varepsilon_2,\, \varepsilon_6,\varepsilon_7\}\triangleq \varepsilon_{\min}, 
\end{equation}
where $\varepsilon_1$, $\varepsilon_2$, $\varepsilon_6$ and $\varepsilon_7$ are defined in~\eqref{v1},~\eqref{v2},~\eqref{v6} and~\eqref{v7}, then all local minima within $\mathcal{S}$ are stable. \end{proof}

By summarizing Lemmas~\ref{away}-\ref{localmini}, we can derive the following convergence result:
\begin{theorem}\label{convergence}
Assume that $G(T_s)$ is connected and $\varepsilon<\varepsilon_{\min}$ by~\eqref{v8}. 
Then starting from anywhere in the workspace except a set of measure zero,  there exists a finite time $T_f\in [T_s,\infty)$ and one agent $i^\star\in \mathcal{N}_a$, such that $x_j(T_f)\in \pi_{i^\star\text{g}}$, $\forall j\in \mathcal{N}$, while at the same time  $\|x_i(t)-x_j(t)\|<r$, $\forall (i,\,j)\in E(T_s)$ and $\forall t\in [T_s,\, T_f]$.
\end{theorem}
\begin{proof}

First of all, the second part follows directly from Theorem~\ref{connectivity} which guarantees that all   edges within $E(T_s)$ will be reserved for all $t>T_s$. Secondly, we have shown that $V(t)$  by~\eqref{Lyapunov} is non-increasing for all $t>T_s$ by Theorem~\ref{connectivity}. By LaSalle's invariance principle~\cite{khalil2002nonlinear} we only need to find out the largest invariant set within $\dot{V}(t)=0$.
Theorems~\ref{away} and~\ref{localmini} ensure that the potential function $V(t)$ has only  local minima inside $\mathcal{S}$ and  saddle points outside $\mathcal{S}$. These saddle points have attractors of measure zero by Lemma~\ref{isolated}.
Thus starting from anywhere in the workspace except a set of measure zero, the system converges to the set of local minima.
Part~(I) of Lemma~\ref{distance} shows that a local minimum can not belong to two different $\mathcal{S}_i$ simultaneously. 
Thus the system converges to the set of local minima within $\mathcal{S}_{i^\star}$ for one active agent $i^\star\in \mathcal{N}_a$.
By part~(II) of Lemma~\ref{distance}, all agents would enter $\pi_{i^\star\text{g}}$, i.e., $x_j\in \pi_{i^\star\text{g}}$, $\forall j\in \mathcal{N}$.
Consequently, there exists $T_f<\infty$ that $x_j(T_f)\in \pi_{i^\star\text{g}}$, $\forall j\in \mathcal{N}$, for exactly one active agent $i^\star\in \mathcal{N}_a$.
\end{proof}
\begin{remark}\label{general}
Note that Theorem~\ref{convergence} holds for any number of active agents with $1\leq N_a\leq N$. In other words, independent of the number of active agents within the team, one active agent will reach its goal region first within finite time, while fulfilling the relative-distance constraints. 
\end{remark}

\subsection{Hybrid Control Structure}

In {Sec.~\ref{synthesis}, we have generated a sequence of progressive goal regions for each agent and in} Sec.~\ref{continuous-design}, we have shown that under the proposed control laws all agents converge to one active agent's {progressive} goal region. In this part, we {propose a local  procedure for each agent to decide on its own activity/passivity. Thus we integrate the discrete plans and the continuous control laws into
a hybrid control scheme to guarantee that every agent's local task is fulfilled. }

\subsubsection{Switching Protocol for sc-LTL}\label{switch}

{Let us first focus on the case when each task $\varphi_i$, $i \in \mathcal N$ is an sc-LTL formula.} As {introduced} in Sec.~\ref{synthesis}, the {discrete} plan $\tau_i $ for agent $i$ {can be represented} by a finite {satisfying prefix of progressive} goal regions in $\Pi_i$ of length $k_i >0$: 
{
$$
\tau_{i,\textup{pre}} = (\pi_{i1},w_{i1})\cdots(\pi_{i{k_i}},w_{ik_i})
$$} 
{We propose the following \emph{activity switching protocol} for each agent $i\in \mathcal{N}$}:

\begin{itemize}\itemsep-0.5ex
\item[(I)] At time $t=0$, agent $i$ {sets $\varkappa_i:= 1$ and itself as} active and sets {$\pi_{i\text{g}}:= \pi_{i\varkappa_i}$}, namely the first goal region in $\tau_i$. The \emph{active} controller~\eqref{law1} is applied, where the progressive goal region is $\pi_{i\text{g}}$.

\item[(II)] Whenever agent $i$ reaches its current {progressive} goal region $\pi_{ig} =  \pi_{i\varkappa_i}$ and $\varkappa_i<k_i$, it provides the  {prescribed set of services $w_{i\varkappa_i}$} and it sets $\varkappa_i:= \varkappa_i+1$ and $\pi_{i\text{g}}:= \pi_{i\varkappa_i}$. {The controller~\eqref{law1} for agent~$i$ is updated while 
the other agents' controllers remain unchanged. }

\item[(III)] Whenever agent $i$ reaches {its last progressive} goal region $\pi_{ig}= \pi_{i{k_i}}${, it provides the set of services $w_{ik_i}$ by which it finishes the execution of its discrete plan.}
Afterwards it remains \emph{passive} and controller~\eqref{law2} applies. 
\end{itemize}

\begin{theorem}\label{satisfyall}
{By following the protocol above, it is guaranteed that $\forall i \in \mathcal N$, $\varphi_i$ is satisfied by $\mathbf{x}_i(T)$, and $\|x_i(t)-x_j(t)\|<r$, $\forall (i,\,j)\in E(0)$ and $\forall t \geq 0$, where $T\rightarrow\infty$.}
\end{theorem}
\begin{proof}
At $t=0$, all agents are active and following the controller~\eqref{law1}. By Theorem~\ref{convergence}, all agents converge to one agent's goal region at a finite time $t_1>0$. Denote by $i \in \mathcal{N}$ this agent. Then either by step~(II) of the protocol, the agent $i$ updates its active control law by setting $\pi_{i\text{g}}=\pi_{i2}$, or by step~(III) the agent $i$ has completed its plan $\tau_{i,\text{pre}}$ and becomes passive.
{Since all agents' plans are finite and Theorem~\ref{convergence} holds for any number of active agents, we obtain that there exists a finite time instant $T_{f_j}$, at which one of the agents $j \in \mathcal N_a$ finishes executing its plan $\tau_{j,\mathrm{pre}}$, i.e., such that $\varphi_{j}$ becomes satisfied.}
Then by step~(III), this agent is passive and following the controller~\eqref{law2} for all times $t \in [T_{f_j}, \infty)$.
{Inductively, we conclude that there exists a time instant $T_f$, by which all agents complete their plans and all formulas are satisfied. All agents are passive for all $t \in (T_{f},\infty)$
%
and by controller~\eqref{law2} they all converge to one point.}
The second part of the theorem follows directly from Theorem~\ref{convergence}. 
\end{proof}

Note that this protocol is fully decentralized as the decisions on an agents' activity/passivity are local and do not depend on any relative-state measurements.


\subsubsection{Switching Protocol for full LTL}\label{switch-infinite}
As {introduced} in Sec.~\ref{synthesis}, if the {task specification $\varphi_i$ is given as a general LTL formula, then the plan $\tau_i$ is represented} by an infinite sequence of {progressive goal regions in a prefix-suffix form}
\begin{align*}
\tau_{i}  =  \tau_{i,\text{pre}}(\tau_{i,\text{suf}})^{\omega} & = (\pi_{i1},w_{i1})(\pi_{i2},w_{i2})\ldots, \text{where} \\
 \tau_{i,\text{pre}}  &  = (\pi_{i},w_{i1})\ldots(\pi_{i{k_i}},w_{ik_i}), k_i>0 \text{ and}          \\ 
\tau_{i,\text{suf}} & = (\pi_{i{k_i+1}},w_{ik_i+1})\ldots(\pi_{i{K_i}},w_{iK_i}), K_i>0.
\end{align*}

The main challenge in this case is to ensure that each agent visits its progressive goal region infinitely often. The activity switching protocol from Sec.~\ref{switch} could not be applied here since all agents would remain active at all times. As a result, the team may repetitively converge to $\pi_{ig}$ for some $i \in \mathcal N$ while never visiting the other agents' progressive goal regions
 (see Sec.~\ref{sec:example} for an example). Hence, we aim to design a ``fair'' activity switching protocol that enforces a progress towards each agent's task.}
{Thereto, we first introduce a communication-free reaching-event detector that enables an agent to monitor its neighbors' plan executions.}

\medskip

\noindent{\emph{Reaching-Event Detector}}. 
 Agent $i \in \mathcal{N}$ can detect when it reaches its own {progressive} goal region $\pi_{i\text{g}}$ by checking if $x_i(t)\in \pi_{i\text{g}}$. 
For our switching protocol presented below, it is also essential that it can detect when another agent
 $j\in \mathcal{N}$ reaches $\pi_{jg}$. Note that by Lemma~\ref{difference}, the connectivity graph is complete since the first time any agent $i \in \mathcal N$ reaches its progressive goal region $\pi_{ig}$, hence it is sufficient to detect when a neighboring agent $j \in \mathcal N_i(t)$ reaches $\pi_{jg}$.
 

Given that the agents satisfy the dynamics by~\eqref{dynamics} and that each agent $i\in \mathcal{N}$ can measure $x_i(t)-x_j(t)$ for all its neighbors $j\in \mathcal{N}_i(t)$ in real time, we assume that  the agent~$i$ can measure or estimate $u_j(t)$, for all $j\in \mathcal{N}_i(t)$~\cite{Fran08}. 
{Let $\Omega_i(j,\,t)\in \mathbb{B}$ be a Boolean variable indicating that agent~$i$ detects its neighboring agent $j\in \mathcal{N}_i(t)$ reaching the goal region $\pi_{j\text{g}}$ at time $t>0$. We propose the following reaching-event detector inspired by~\cite{Tabu11}. Simply speaking, the detector checks if within a short time period $[t-\Delta_t,\, t]$, there exists $j \in \mathcal N_i(t)$, such that $u_j(t)$ has changed from a relatively small value (below a given $\Delta_u$) by a difference larger than certain $\Delta_d$. If so, it means that the agent $j$ has reached its progressive goal region $\pi_{j\text{g}}$. 

The choice of this reaching-event detector is motivated by the following facts: 
By~\eqref{equi}, all control inputs $u_i(t)$ are close to zero when the system is close to a local minimal, $\forall i\in \mathcal{N}$. 
Afterwards, our switching protocol introduced below guarantees that
\emph{only} agent $j$ switches its control law  either to~\eqref{law1} in order to navigate to the next progressive goal region or to~\eqref{law2} in order to become passive. This change is lower-bounded by constant $\Delta_d$ derived using~control law \eqref{law1} and Lemmas~\ref{distance},~\ref{furtherbound} as
$\Delta_d\triangleq |f(r_{\min})-f(\sqrt{0.4\varepsilon})|$, where $f(\|p_j\|)=d_j(\|p_j\|)\|p_j\|$ is a scalar function and $d_j(\|p_j\|)$ is defined by~\eqref{d}.
In contrast, for the other agents $i\neq j$, $i \in \mathcal{N}$, the control input $u_i(t)$ remains unchanged and close to zero. Hence, agent $j$ is identified as the only one who has reached its progressive goal region.
Formally,

\begin{definition}
$\Omega_i(j,\,t) \define \textup{\texttt{True}}$
if and only if there exists $ t'\in [t-\Delta_t,\, t]$, where $|u_j(t')|<\Delta_u$ and $|u_j(t)-u_j(t')|>\Delta_d$.

\end{definition}


\medskip

\noindent\emph{Activity Switching Protocol}. 
Loosely speaking, in the proposed protocol, an agent $i \in \mathcal N$ becomes passive if it has made a certain progress towards the satisfaction of its specification, hence giving the other agents an opportunity to advance in execution of their plans. However, once each agent has achieved certain progress, the agent $i$ becomes active again to proceed with its infinite plan.
We define a \emph{round} as the time period during which each agent has reached at least one of its goal regions according to their plans. 


\begin{definition}

For all $m \geq 1$, the \emph{$m$-th round} is defined as the time interval $[T_{{\circlearrowleft}_{m-1}},T_{{\circlearrowleft}_{m}})$, where $T_{{\circlearrowleft}_{0}} = 0$, $T_{{\circlearrowleft}_{m-1}} < T_{{\circlearrowleft}_{m}}$ and for all $m \geq 1$, $\round$ is the smallest time satisfying the following conditions $\forall i \in \mathcal N$: $\mathtt{word}_i(T_{\circlearrowleft_{m}}) = w_{i1}\cdots w_{i\ell}$ for some $\ell \geq 1$ and
$\mathtt{word}_i(T_{\circlearrowleft_{m}}) \neq \mathtt{word}_i(T_{\circlearrowleft_{m-1}})$.
\end{definition}


The notion of a round is crucial to the design of the activity switching protocol. To recognize a round completion, we introduce the following variables:
$\chi_i\geq 0$ indicates the starting time of the current round, and
$\Upsilon_i \in \mathbb{Z}^N$ is a vector to record how many progressive goal regions each agent has reached within one round since $\chi_i$.
Although these variables are  locally maintained by each agent. 
By Lemma~\ref{difference}, the connectivity graph is complete since the first time one active agent reaches its  goal region $\pi_{ig}$ and  under the assumption of unbiased measurements, it holds that at the same time instant $\chi_i = \chi_j$, and $\Upsilon_i = \Upsilon_j$,  $\forall i,j \in \mathcal N$.
We propose the following \emph{activity switching protocol} for each agent $i\in \mathcal{N}$:

\begin{itemize}\itemsep-0.5ex
\item[(I)] At time  $t=0$, $\Upsilon_i := \textbf{0}_{N}$, $\chi_i : =0$, $\varkappa_i := 1$. 
The agent $i$ is active and follows control law~\eqref{law1}, where $\pi_{i\text{g}} := \pi_{i\varkappa_i}$. 

\item[(II)]
Whenever the agent $i$ reaches its current progressive goal region $\pi_{ig} = \pi_{i\varkappa_i}$ and waits until $|u_i(t)|<\Delta_u$, it provides the prescribed set of services $w_{i\varkappa_i}$ and updates the current progressive goal region accordingly: If $\varkappa_i < K_i$ then $\varkappa_i := \varkappa_i+1$, and if $\varkappa_i = K_i$ then $\varkappa_i := k_i+1$. Furthermore,  $\pi_{ig} := \pi_{i\varkappa_i}$, and finally $\Upsilon_i[i] := \Upsilon_i[i]+1$.

Generally speaking, the agent $i$ decides to stay {active} or to become {passive} based on the probability function:
$$
\textbf{Pr}(b_i=1)= \begin{cases}
f_{\prob}(\cdot) &  \text{if } f_{\cond}(\cdot) = \texttt{True},\\
0& \text{otherwise,}
\end{cases}
$$
where $f_\prob(\cdot) \in [0,1]$ and $f_\cond(\cdot) \in \{\texttt{True},\texttt{False}\}$ are functions of time $t$ and the local variables $\Upsilon_i$ and~$\chi_i$, subject to the following: given that the current round is the $m$-th one, there exists a time $T \in  (T_{\circlearrowleft_{m-1}},\, T_{\circlearrowleft_{m}})$, such that $f_\cond(\cdot)= \texttt{False}$ for all $t \in [T,\, \round)$.

Whenever $b_i=1$, the agent $i$ keeps following the  control law~\eqref{law1} with the updated $\pi_{ig}$. Otherwise, it becomes passive and the control law~\eqref{law2} is applied.


\item[(III)] Whenever agent $i$ detects that $\Omega_i(j,t)=\texttt{True}$, for some $j \neq i \in \mathcal N$, it sets $\Upsilon_i[j]=\Upsilon_i[j]+1$.

\item[(IV)] Whenever for all $j \in \mathcal N$ it holds that $\Upsilon_i[j]>0$ the agent $i$ sets $\Upsilon_i:=\textbf{0}_{N}$, $\chi_i:=t$ and follows the active control law~\eqref{law1}.  
 \end{itemize}

{A straightforward example of the functions choice in (II) is $f_\cond = \texttt{False}$, for all $t \geq 0$. Then the agent $i$ always becomes passive once it visits $\pi_{ig}$. Note that it becomes active after the current round is completed by step (IV). However, a different selection of the functions may allow for trading the fairness of activity switching for  the increased efficiency of plan executions.
The switching to passive control mode may be temporarily postponed and as a result, the visits to progressive goal regions may become more frequent. An example of such a case is given in Sec.~\ref{sec:example}. }


\begin{lemma}\label{finitetime}
The round $[T_{\circlearrowleft_{m-1}},\round)$ is finite, $\forall m \geq 1$.
\end{lemma}
\begin{proof}
Let $t=T_{\circlearrowleft_{m-1}} = 0$, and thus $\Upsilon_i[j]=0$, for all $i,j \in \mathcal N$ by step (I). By Theorem~\ref{convergence}, one of the agents reaches its progressive goal region in finite time at $t_1 \geq T_{\circlearrowleft_{j-1}}$. Since there are only finite number of agents and due to the required properties of $f_\cond$, there exists a finite time $T_{f_j} \geq 0$, when either the step (IV) applies or when one of the agents $j \in \mathcal N_a$ necessarily becomes passive by the function $\textbf{Pr}(\cdot)$ in step (II) and remains passive till the end of the round. In the former case $\round=T_{f_j}$, i.e., we directly obtain that the $m$-th round is finite. In the latter case, by inductive reasoning we obtain that there exists a finite time instant $T_f$, such that step (IV) applies, i.e., such that $\round=T_{f}$. Again, we have that $m$-th round is finite. 

Inductively, let $m >1$, $t = T_{\circlearrowleft_{m-1}}$, and $\Upsilon_i[j]=0$, for all $i,j \in \mathcal N$ by step (IV). Using analogous arguments as above, we obtain the existence of a finite $ T{\circlearrowleft_{m}}$.
\end{proof}

\begin{theorem}
By following the  protocol above, it is guaranteed that $\forall i\in \mathcal{N}$, $\varphi_i$ is satisfied by $\mathbf{x}_i(T)$ and $\|x_i(t)-x_j(t)\|<r$, $\forall (i,\,j)\in E(0)$ and $\forall t>0$, where $T \rightarrow \infty$.
\end{theorem}
\begin{proof}
The satisfaction of $\varphi_i$ follows directly from the correctness of each agent's discrete plan and the fact that each round is finite by Lemma~\ref{finitetime}. 
At last, the distance constraints between neighbouring agents are always maintained 
as shown in Theorem~\ref{convergence}.
\end{proof}

\section{Simulation}\label{sec:example}
In the following case study, we simulate a team of four autonomous robots $\mathcal N = \{\mathfrak{R}_1, \cdots, \mathfrak{R}_4\}$ {subject to the dynamics~\eqref{dynamics} in a bounded, obstacle-free workspace of $40 \times 40$ meters ($m$)}. {Each robot $\mathfrak{R}_i$ is given a local sc-LTL or LTL task $\varphi_i$}. All algorithms and modules {were} implemented in Python 2.7. {Simulations were} carried out on a desktop computer (3.06 GHz Duo CPU and 8GB of RAM) with a simulation stepsize set to $1ms$.  

\begin{figure}[t!]
\begin{minipage}[t!]{0.49\linewidth}
\centering
\includegraphics[width =1\textwidth, height=1\textwidth]{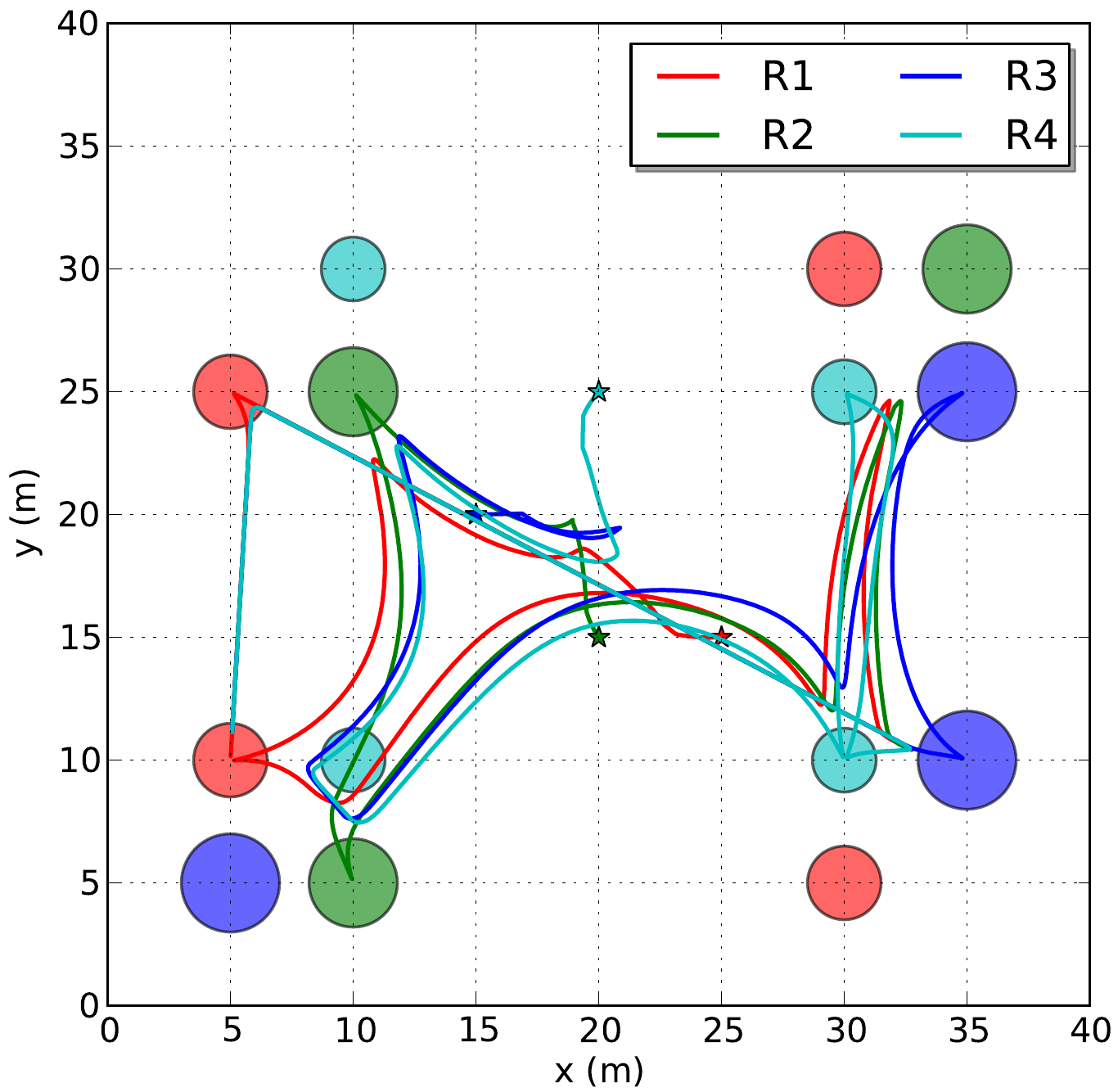}
 \end{minipage}
\begin{minipage}[ht!]{0.49\linewidth}
\centering
   \includegraphics[width =1\textwidth, height=0.54\textwidth]{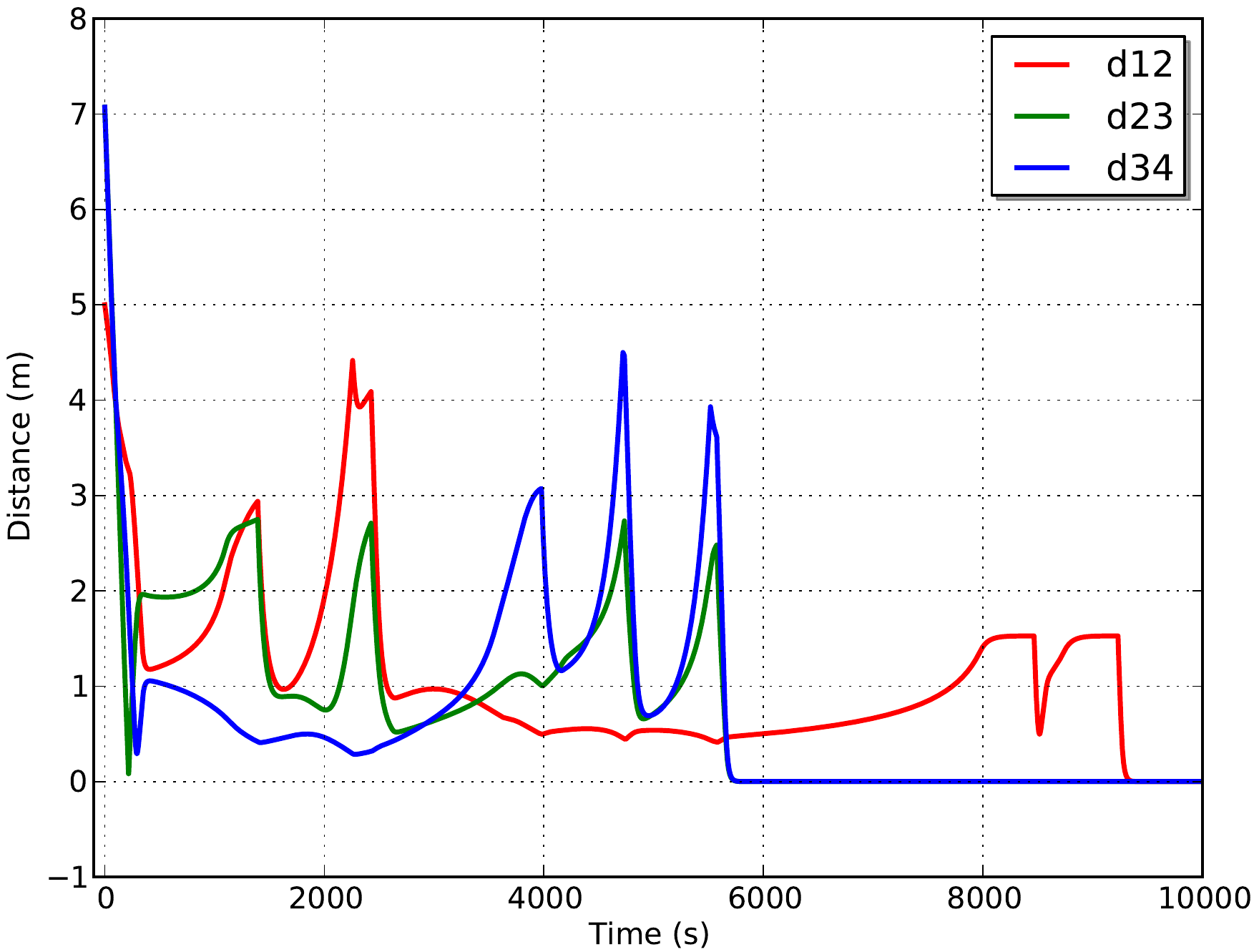}
   \includegraphics[width =0.95\textwidth, height=0.44\textwidth]{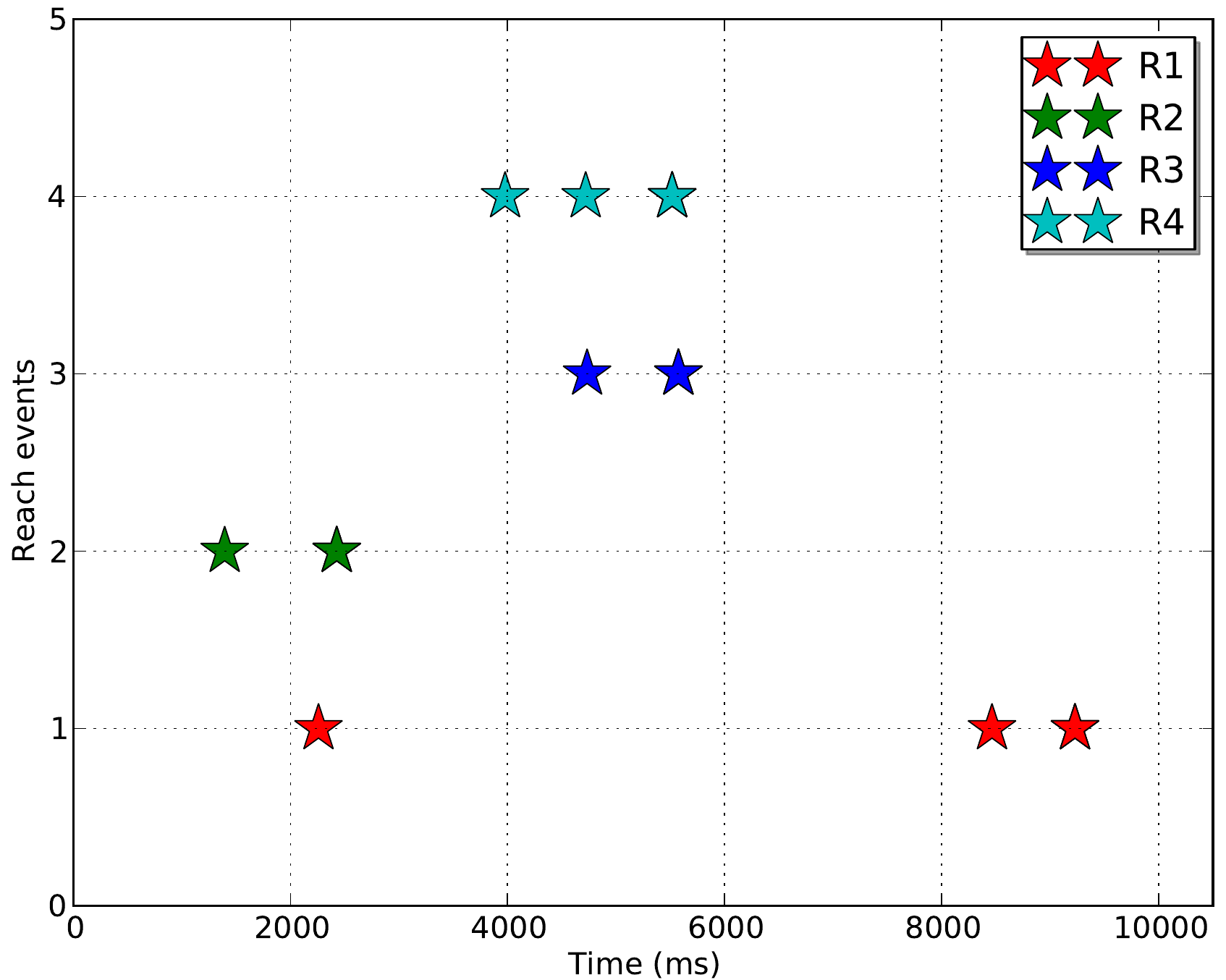}
\end{minipage}
\caption{Left: agents' respective regions of interest in red, green, blue and cyan respectively and their trajectories after execution of switching policy from Sec.~\ref{switch}. All agents accomplish their sc-LTL tasks after $9s$. Top-right: the evolution of pair-wise distances $\|x_{12}\|,\|x_{23}\|, \|x_{34}\|$, which all stay below $7.5m$. Bottom-right: the time instants when the agents reach their goal regions.}
\label{safe_static}
\end{figure}
\medskip
As shown Fig.~\ref{safe_static}, several sphere regions of interest for each agent are placed in top-left, top-right, bottom-right, and bottom-left corners of the workspace and they all satisfy Assump.~\ref{region-assump} with $c_{\max}=40$ and $r_{min}=2$:
\begin{itemize}\itemsep-0.5ex
\item  $\Pi_1 = \{\varpi_{1\tl},\varpi_{1\tr},\varpi_{1\br},\varpi_{1\bl}\}$  shown in red;
\item  $\Pi_2 = \{\varpi_{2\tl}, \varpi_{2\tr}, \varpi_{2\bl}\}$  shown in green; 
\item $\Pi_3 = \{\varpi_{3\tr}, \varpi_{3\br}, \varpi_{3\bl}\}$ shown in blue;
\item $\Pi_4 = \{\varpi_{4\tl}, \varpi_{4\tr}, \varpi_{4\br}, \varpi_{4\bl}\}$ shown in cyan.
\end{itemize}

The respective sets of atomic propositions (services) are $\AP_1 = \{\sigma_{11},\sigma_{12}\}$; $\AP_2 = \{\sigma_{21},\sigma_{22}, \sigma_{23}\}$; $\AP_3 = \{\sigma_{31},\sigma_{32}, \sigma_{33}\}$; and $\AP_4 = \{\sigma_{41},\sigma_{42}\}$. The regions are labeled as follows:
$L_1(\varpi_{1\tl})=L_1(\varpi_{1\br}) = \{\sigma_{11}\}$, $L_1(\varpi_{1\tr})=L_1(\varpi_{1\bl}) = \{\sigma_{12}\}$;  
$L_2(\varpi_{2\tl})= \{\sigma_{21}\}$,  $L_2(\sigma_{2\tr})=\{\sigma_{22}\}$, $L_2(\varpi_{2\bl})=\{\sigma_{23}\}$; 
$L_3(\varpi_{3\tr})= \{\sigma_{31}\}$,  $L_3(\varpi_{3\br})=\{\sigma_{32}\}$, $L_3(\varpi_{3\bl})=\{\sigma_{33}\}$; and finally
$L_4(\varpi_{4\tl})=L_4(\varpi_{4\tr}) = \{\sigma_{41}\}$, $L_4(\varpi_{4\bl})=L_4(\varpi_{4\br}) = \{\sigma_{42}\}$.
The agents start from $[25,15]$, $[20,15]$, $[15,20]$, and $[20,25]$, respectively. 
The uniform neighboring radius is $r=8m$ and the design parameter needed in Def.~\ref{edge} is $\delta=0.5m$.
The edge set of $G(0)$ is hence $E(0)=\{(\mathfrak{R}_1,\mathfrak{R}_2),(\mathfrak{R}_2,\mathfrak{R}_3)$, $(\mathfrak{R}_3,\mathfrak{R}_4)\}$. 
The upper bound by~\eqref{v8} is $\varepsilon < \varepsilon_{\textrm{min}} \approx 0.031$ and we choose $\varepsilon=0.03$. 


We consider two  cases of the agent task specifications: one {with} sc-LTL formulas and one {with} general LTL formulas.

\medskip

\noindent
\emph{sc-LTL Task Specifications.}
The local task of agent $\mathfrak{R}_1$ to provide service $\sigma_{12}$, then $\sigma_{11}$ and at last again $\sigma_{12}$. The corresponding LTL formula is $\varphi_1 = \Diamond(\sigma_{12} \wedge \Diamond (\sigma_{11} \wedge \Diamond \sigma_{12}))$.
Agent $\mathfrak{R}_2$ is asked to provide service $\sigma_{21}$ or $\sigma_{22}$ and service $\sigma_{23}$ in any order, formalized as $\varphi_2^s= \Diamond(\sigma_{21} \vee \sigma_{22})\wedge \Diamond \sigma_{23}$.  
The task of agent $\mathfrak{R}_3$ is to provide service $\sigma_{31}$ or $\sigma_{32}$ and service $\sigma_{33}$ in any order, formalized as $\varphi_2^s= \Diamond(\sigma_{31} \vee \sigma_{32})\wedge \Diamond \sigma_{33}$. 
Finally, agent $\mathfrak{R}_4$ is required to provide service $\sigma_{42}$, then $\sigma_{41}$ and at last again service $\sigma_{42}$, represented by the LTL formula $\varphi_4 =\Diamond(\sigma_{42} \wedge \Diamond (\sigma_{41} \wedge \Diamond \sigma_{42}))$.

The synthesized discrete plans are as follows: 
\begin{itemize}\itemsep-0.5ex
\item $\tau_{1} = (\varpi_{1\bl},\{\sigma_{12}\})(\varpi_{1\tl},\{\sigma_{11}\})(\varpi_{1\bl},\{\sigma_{12}\})$
\item $\tau_{2} = (\varpi_{2\tl},\{\sigma_{21}\})(\varpi_{2\bl},\{\sigma_{23}\})$
\item $\tau_{3} = (\varpi_{3\tr},\{\sigma_{31}\})(\varpi_{3\br},\{\sigma_{33}\})$
\item $\tau_{4} = (\varpi_{4\br},\{\sigma_{41}\})(\varpi_{4\tr},\{\sigma_{42}\})(\varpi_{4\tl},\{\sigma_{41}\})$
\end{itemize}

At $t=0$, the switching policy from Sec.~\ref{switch} is applied.  
The agent trajectories are shown in Fig.~\ref{safe_static}, where the distances between the   neighboring agents along with times of reaching the agents' progressive goal regions are  plotted, too. 



\medskip
\begin{figure}[t]
\begin{minipage}[t!]{0.49\linewidth}
\centering
\includegraphics[width =1\textwidth, height=1\textwidth]{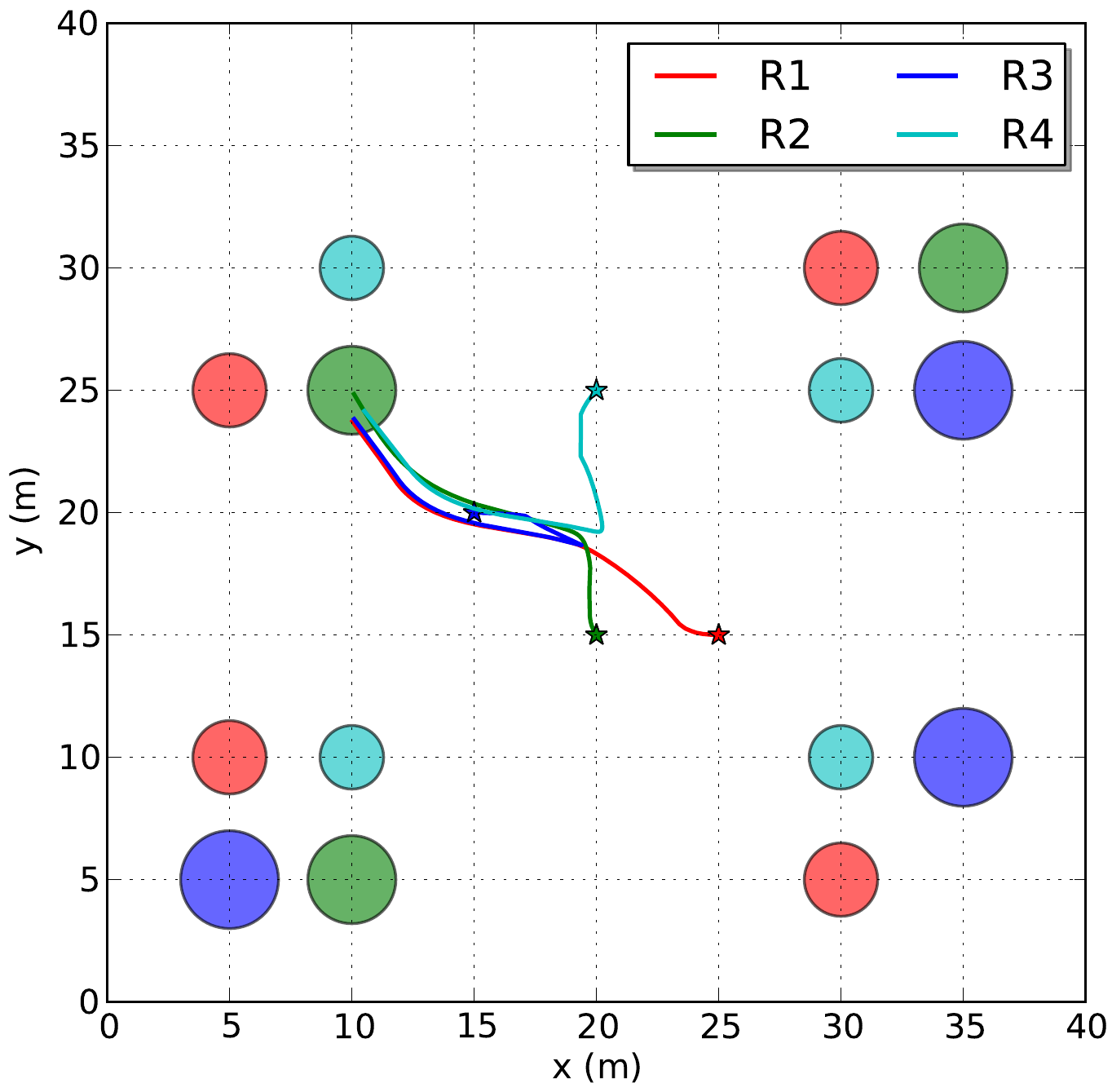}
 \end{minipage}
\begin{minipage}[ht!]{0.49\linewidth}
\centering
   \includegraphics[width =1\textwidth, height=0.54\textwidth]{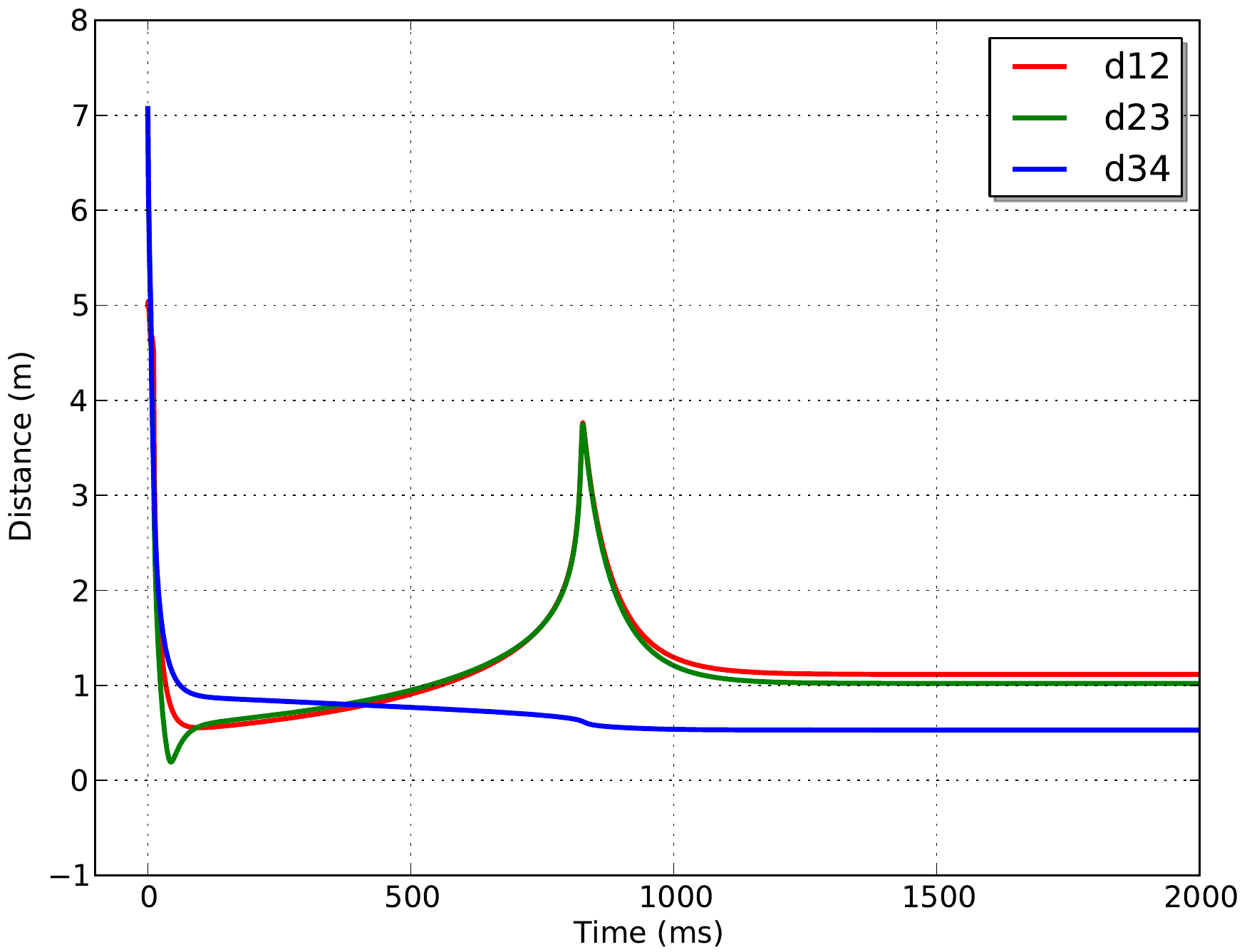}
   \includegraphics[width =0.95\textwidth, height=0.44\textwidth]{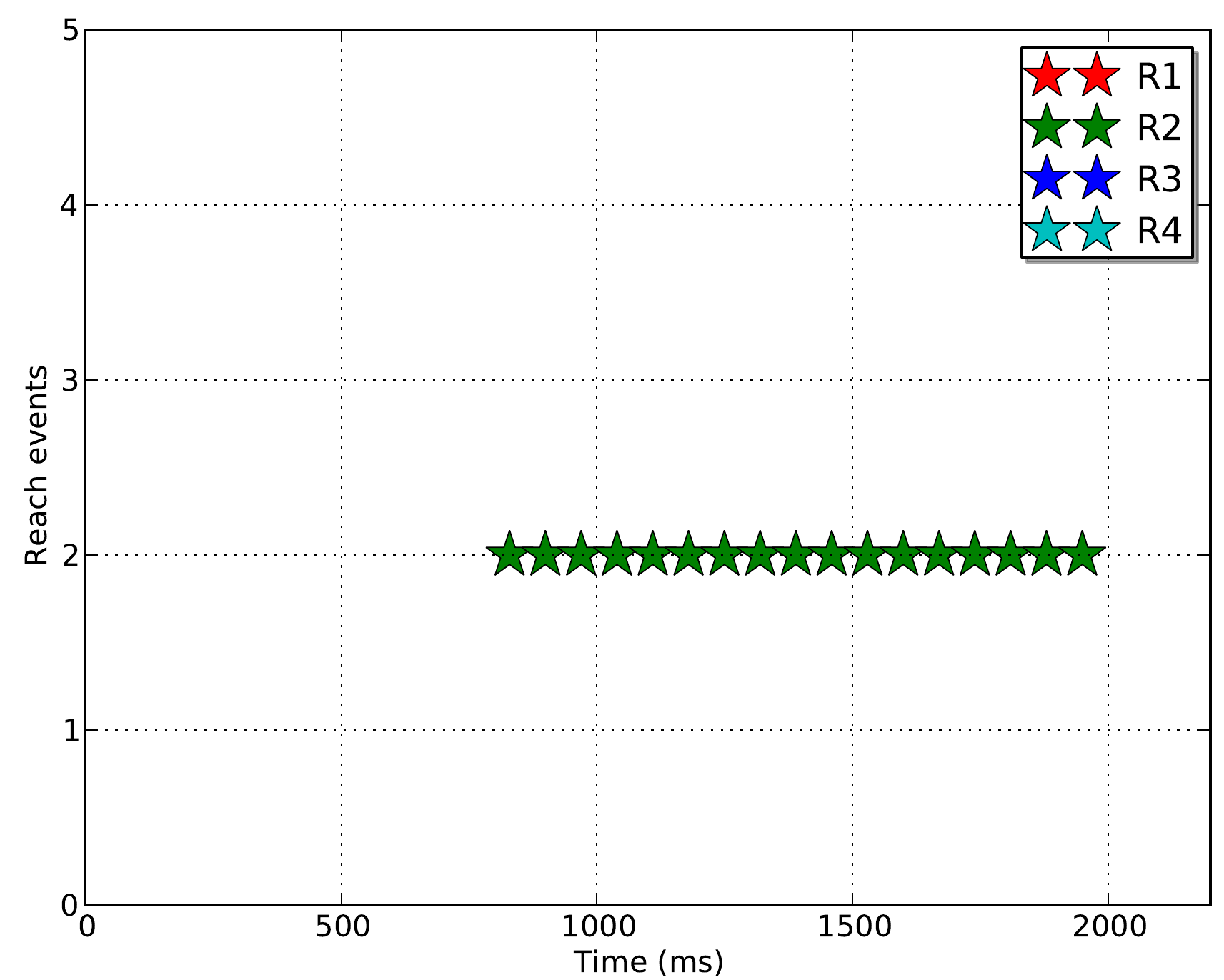}
\end{minipage}
\caption{The system with general LTL formulas after application of switching policy from Sec.~\ref{switch}. All agent converge to $\varpi_{2\tl}$ and stop there due to the fact that $\mathfrak{R}_2$ remains active with  an infinite plan to stay at  $\varpi_{2\tl}$.}
\label{live_stuck}
\end{figure}

\noindent \emph{General LTL task specifications}. 
The task of agent $\mathfrak{R}_1$ to periodically provide both services $\sigma_{11}$ and $\sigma_{12}$, represented by $\phi_1 = \square \Diamond \sigma_{11} \wedge \square \Diamond \sigma_{12}$. 
The task of agent $\mathfrak{R}_2$ is to periodically provide one of the services $\sigma_{21}$ or $\sigma_{22}$ or $\sigma_{23}$, formalized as $\phi_2= \square \Diamond(\sigma_{21} \vee \sigma_{22} \vee \sigma_{23})$. Finally, the tasks of agents $\mathfrak{R}_3$ and $\mathfrak{R}_4$ are 
$\phi_3= \square \Diamond(\sigma_{31} \vee \sigma_{32} \vee \sigma_{33})$, and $\phi_4 = \square \Diamond \sigma_{41} \wedge \square \Diamond \sigma_{42}$.

The synthesized discrete plans are as follows: 
\begin{itemize}\itemsep-0.5ex
\item $\tau_{1} = \big((\varpi_{1\bl},\{\sigma_{12}\})\big((\varpi_{1\tl},\{\sigma_{11}\})\big)^\omega$
\item $\tau_{2} = (\varpi_{2\tl},\{\sigma_{21}\})^\omega$
\item $\tau_{3} = (\varpi_{3\bl},\{\sigma_{33}\})^\omega$
\item $\tau_{4} = \big((\varpi_{4\br},\{\sigma_{41}\})(\varpi_{4\tr},\{\sigma_{42}\})\big)^\omega$
\end{itemize}



{
First, we simulated the scenario where we applied the activity switching protocol for sc-LTL formulas proposed in Sec.~\ref{switch}. Fig.~\ref{live_stuck} shows that the first progressive goal visited is $\varpi_{2\tl}$. Since the agent $\mathfrak{R}_2$ stays active by the protocol, and its next progressive goal region is again $\varpi_{2\tl}$, the whole system has reached its stable local minimum. Hence, all agents converge very close to one point and stop. 
In contrast, the activity switching protocol from Sec.~\ref{switch-infinite} avoids such an unwanted behavior.

The simulation results for the activity switching protocol from Sec.~\ref{switch-infinite} are illustrated in Fig.~\ref{live_static}.
The functions $f_{\prob}$ and $f_{\cond}$ were chosen in a way that allows to partially trade fairness of activity switching for increased efficiency of plan executions measured in terms of the distance traveled between consecutive visits to progressive goal regions. More specifically, an agent is not switched to passive immediately after it reaches one of its goal region. Rather than that, it has the following probability of remaining active:
$$
\textbf{Pr}(b_i=1)= \begin{cases}
e^{-\alpha_i\Upsilon_i[i](t-\chi_i)},&\quad \text{if}\quad \Upsilon_i[i]\cdot (t-\chi_i)<\bar{\chi}_i,\\
0,&\quad \text{if}\quad \Upsilon_i[i]\cdot (t-\chi_i)\geq \bar{\chi}_i,
\end{cases}
$$
where $\bar{\chi}_i = 5$, and $\alpha_i = 1$. The probability of remaining active decreases with the increasing time elapsed since the current round  started and with the increasing number agent $\mathfrak{R}_i$'s own progressive goal region was visited. Note that there exists a finite $T \in (T_{\circlearrowleft_{m-1}},T_{\circlearrowleft_{m}})$, such that $\Upsilon_i[i]\cdot (t-\chi_i)\geq \bar{\chi}_i$ for all $t \in [T, \round]$, hence each agent $\mathfrak{R}_i$ is guaranteed to be switched to passive control mode eventually.

The selected function does not necessarily yield a monotonic decrease of the total number of active agents in the team and is particularly useful when one agent has a set of goal regions whose locations are close.

\begin{figure}[t]
\begin{minipage}[t!]{0.49\linewidth}
\centering
\includegraphics[width =1\textwidth, height=1\textwidth]{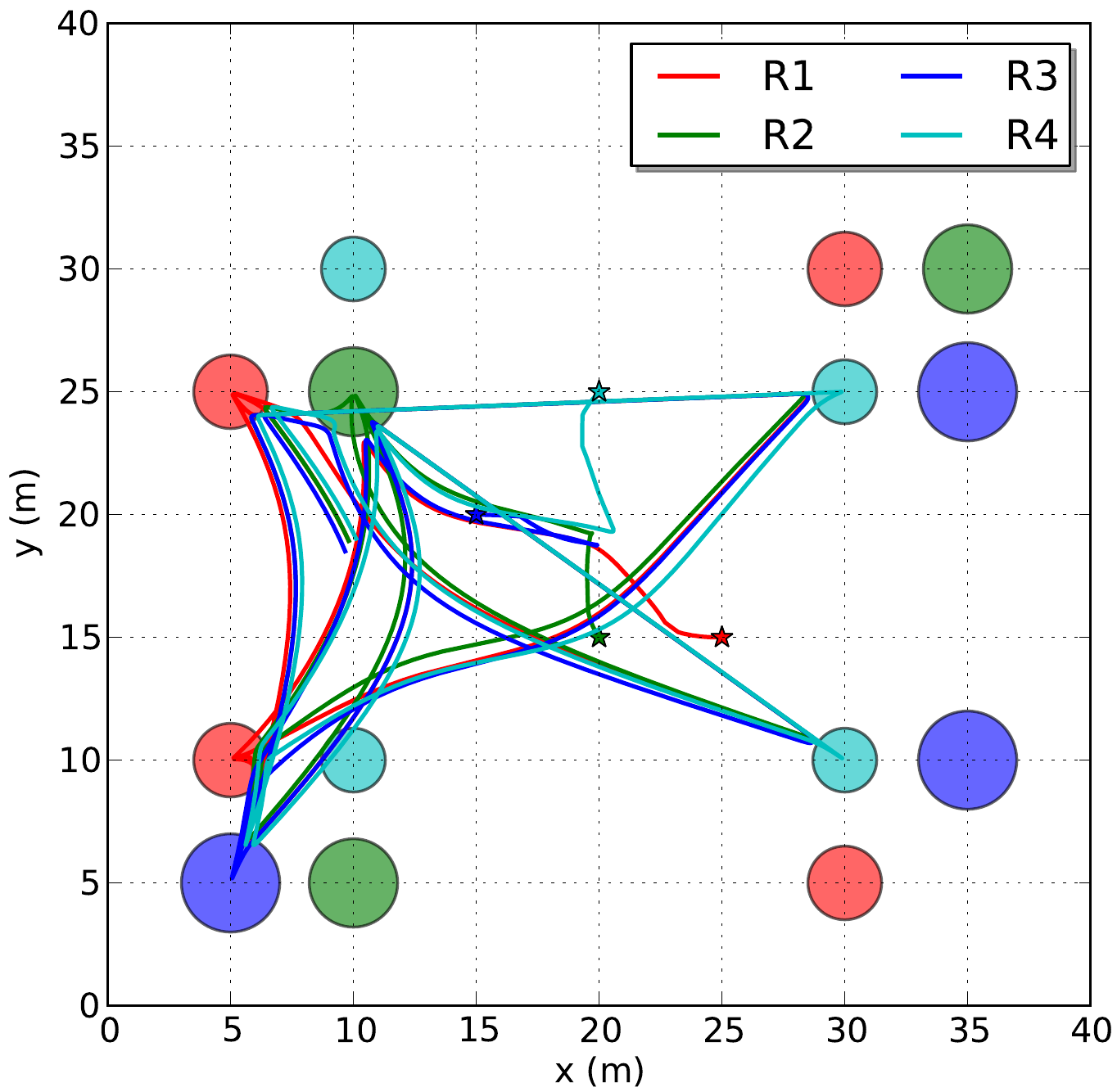}
 \end{minipage}
\begin{minipage}[ht!]{0.49\linewidth}
\centering
   \includegraphics[width =1\textwidth, height=0.54\textwidth]{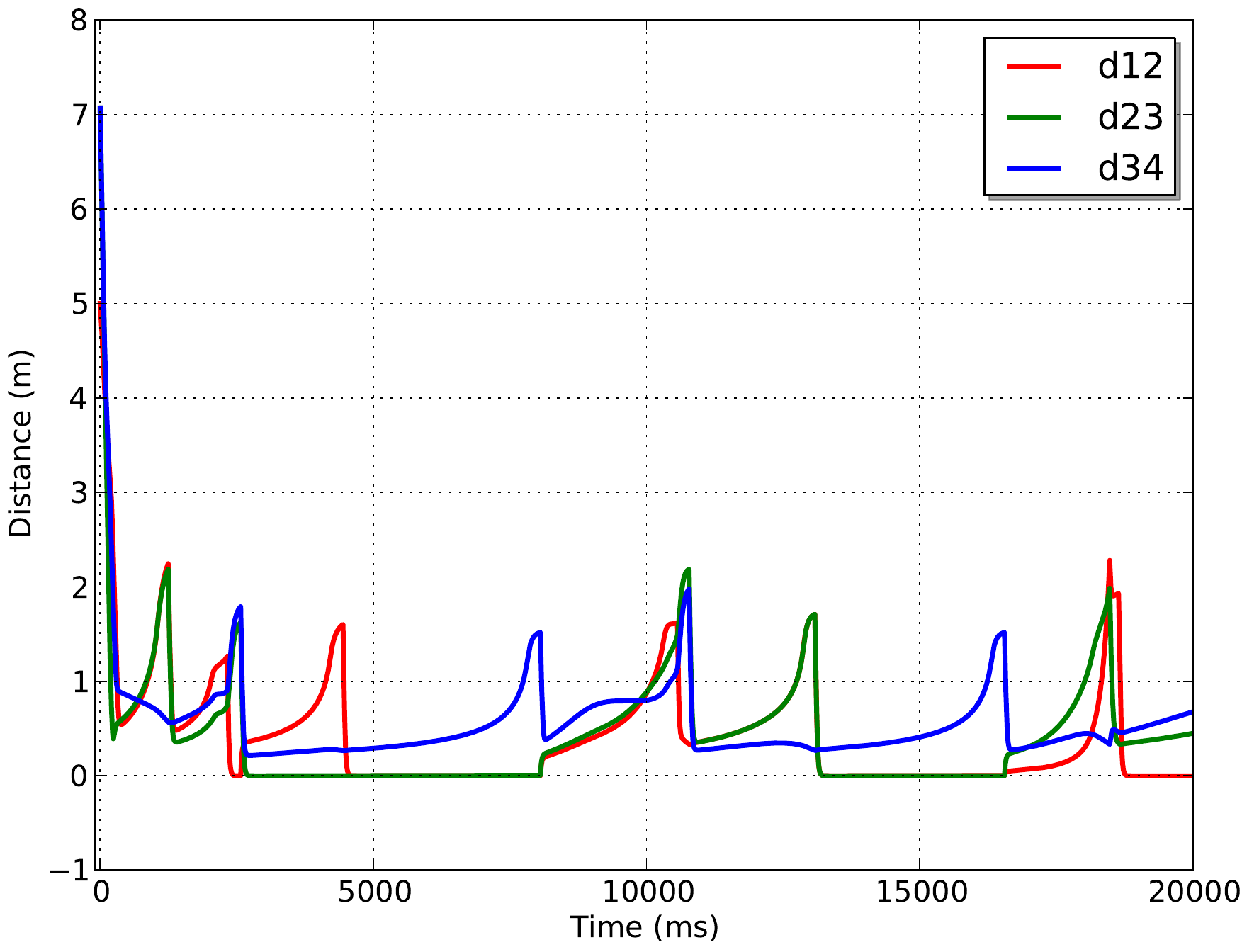}
   \includegraphics[width =0.95\textwidth, height=0.44\textwidth]{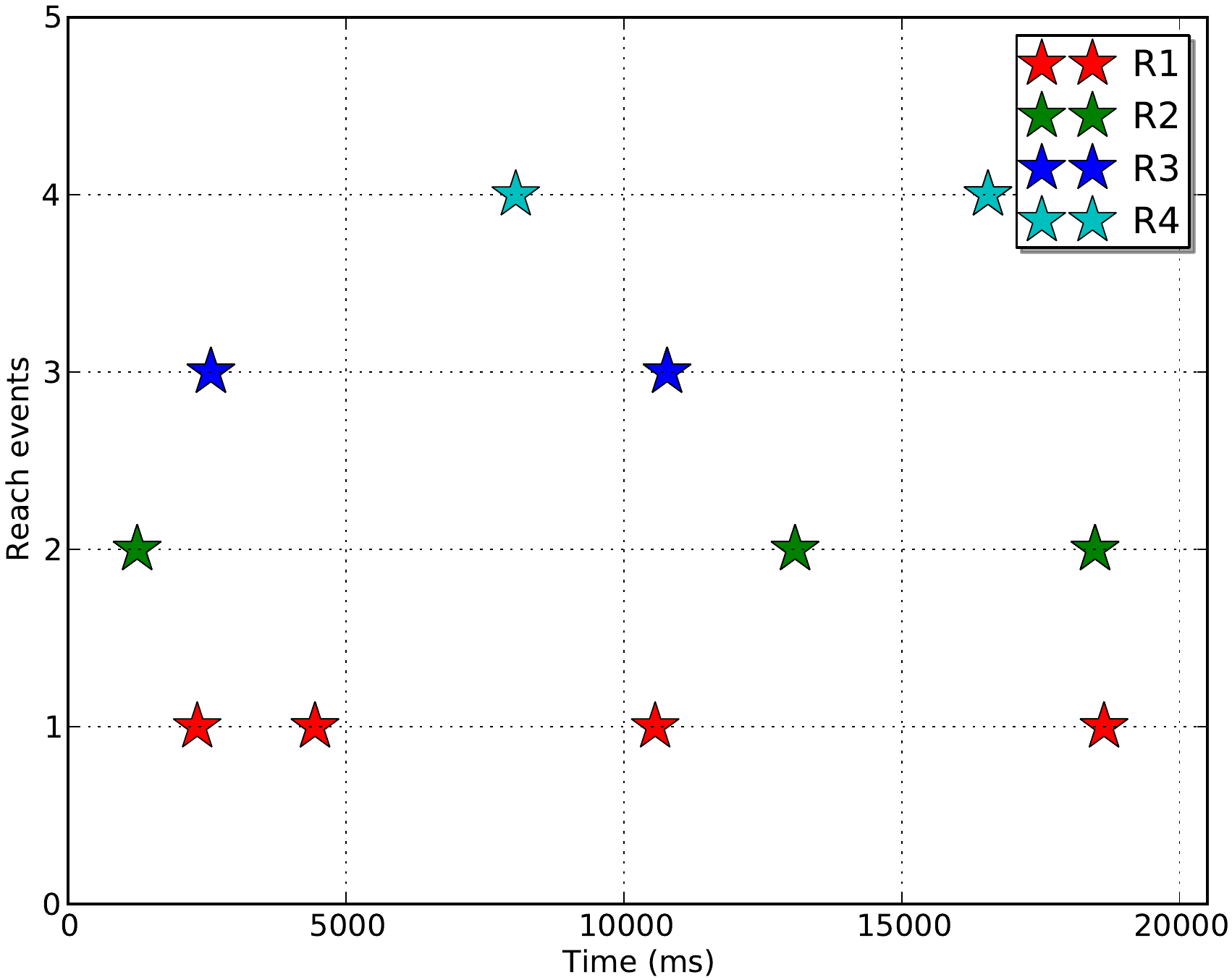}
\end{minipage}
\caption{Left: agents' respective regions of interest in red, green, blue and cyan respectively and their trajectories after execution of switching policy from Sec.~\ref{switch-infinite} for $20$s. Top-right: the evolution of pair-wise distances $\|x_{12}\|,\|x_{23}\|, \|x_{34}\|$, which all stay below $7.5m$. Bottom-right: the time instants when the agents reach their goal regions according to their plan.}
\label{live_static}
\end{figure}

\section{Conclusion and Future Work}\label{sec:conc}
We proposed a distributed communication-free hybrid control scheme for multi-agent systems to fulfil locally-assigned tasks as general or sc-LTL formulas, while at the same time subject to relative-distance constraints. 

Future work plans include handling uncertainties in the relative state measurements and considering more complex agent dynamics. We also plan to relax the requirement on the completness of the graph $G(t)$.

\section{Acknowledgements}
{This work was supported by   EU STREP RECONFIG: FP7-ICT-2011-9-600825 and the Swedish Research Council.}

\end{document}